\documentclass[a4paper,UKenglish]{article}

\usepackage{amsfonts}
\usepackage{amssymb}
\usepackage{amsthm}
\usepackage[auth-lg,affil-it]{authblk}
\usepackage{bbm}
\usepackage{bm}
\usepackage{braket}
\usepackage[font=scriptsize,format=plain,labelfont=bf]{caption}
\usepackage[shortlabels,inline]{enumitem}
\usepackage[T1]{fontenc}
\usepackage[margin=1in,marginparsep=0.5em]{geometry}
\usepackage[utf8]{inputenc}
\usepackage{listings}
\usepackage{mathtools}
\usepackage[numbers,sectionbib,sort&compress]{natbib}
\usepackage{todonotes}
\usepackage{xargs}
\usepackage{xfrac}
\usepackage{xspace}

\usepackage{csquotes}     
\usepackage{hyperref}     

\newcommand{\textcite}{\citet}
\newcommand{\Textcite}{\Citet}

\let\originalleft\left
\let\originalright\right
\renewcommand{\left}{\mathopen{}\mathclose\bgroup\originalleft}
\renewcommand{\right}{\aftergroup\egroup\originalright}

\setlist{noitemsep}
\setlist[enumerate,1]{label=(\alph*)}
\setlist[enumerate,2]{label=(\roman*)}

\newtheorem{theorem}{Theorem}[section]
\newtheorem{lemma}[theorem]{Lemma}

\newtheorem{observation}[theorem]{Observation}

\newtheorem{claim}[theorem]{Claim}

\makeatletter
\let\@pklingPr\Pr
{\catcode`\|=\active
  \xdef\pr{\protect\expandafter\noexpand\csname pr \endcsname}
  \expandafter\gdef\csname pr \endcsname#1{\mathinner
        {\@pklingPr({\mathcode`\|32768\let|\midvert #1})}}
  \xdef\Pr{\protect\expandafter\noexpand\csname Pr \endcsname}
  \expandafter\gdef\csname Pr \endcsname#1{\@pklingPr\left(%
     \ifx\SavedDoubleVert\relax \let\SavedDoubleVert\|\fi
     {\let\|\SetDoubleVert
     \mathcode`\|32768\let|\SetVert
     #1}\right)}
}

{\catcode`\|=\active
  \xdef\ex{\protect\expandafter\noexpand\csname ex \endcsname}
  \expandafter\gdef\csname ex \endcsname#1{\mathinner
        {\mathbb{E}[{\mathcode`\|32768\let|\midvert #1}]}}
  \xdef\Ex{\protect\expandafter\noexpand\csname Ex \endcsname}
  \expandafter\gdef\csname Ex \endcsname#1{\mathbb{E}\left[%
     \ifx\SavedDoubleVert\relax \let\SavedDoubleVert\|\fi
     {\let\|\SetDoubleVert
     \mathcode`\|32768\let|\SetVert
     #1}\right]}
}

{\catcode`\|=\active
  \xdef\var{\protect\expandafter\noexpand\csname var \endcsname}
  \expandafter\gdef\csname var \endcsname#1{\mathinner
        {\operatorname{Var}[{\mathcode`\|32768\let|\midvert #1}]}}
  \xdef\Var{\protect\expandafter\noexpand\csname Var \endcsname}
  \expandafter\gdef\csname Var \endcsname#1{\operatorname{Var}\left[%
     \ifx\SavedDoubleVert\relax \let\SavedDoubleVert\|\fi
     {\let\|\SetDoubleVert
     \mathcode`\|32768\let|\SetVert
     #1}\right]}
}

{\catcode`\|=\active
  \xdef\cov{\protect\expandafter\noexpand\csname cov \endcsname}
  \expandafter\gdef\csname cov \endcsname#1{\mathinner
        {\operatorname{Cov}[{\mathcode`\|32768\let|\midvert #1}]}}
  \xdef\Cov{\protect\expandafter\noexpand\csname Cov \endcsname}
  \expandafter\gdef\csname Cov \endcsname#1{\operatorname{Cov}\left[%
     \ifx\SavedDoubleVert\relax \let\SavedDoubleVert\|\fi
     {\let\|\SetDoubleVert
     \mathcode`\|32768\let|\SetVert
     #1}\right]}
}
\makeatother

\newcommand*{\poly}{\mathsf{poly}}
\newcommand*{\polylog}{\mathsf{polylog}}

\DeclarePairedDelimiter\abs{\lvert}{\rvert}
\DeclarePairedDelimiterX{\norm}[1]{\lVert}{\rVert}{#1}

\DeclarePairedDelimiter\ceil{\lceil}{\rceil}

\DeclarePairedDelimiter\intcc{[}{]}
\DeclarePairedDelimiter\intco{[}{)}
\DeclarePairedDelimiter\intoc{(}{]}
\DeclarePairedDelimiter\intoo{(}{)}

\newcommandx*{\LDAUOmicron}[2][1=@pkling_false]{\mathrm{O}\ifthenelse{\equal{#1}{small}}{\bigl(#2\bigr)}{\left(#2\right)}}
\newcommandx*{\LDAUomicron}[2][1=@pkling_false]{\mathrm{o}\ifthenelse{\equal{#1}{small}}{\bigl(#2\bigr)}{\left(#2\right)}}
\newcommandx*{\LDAUOmega}[2][1=@pkling_false]{\Omega\ifthenelse{\equal{#1}{small}}{\bigl(#2\bigr)}{\left(#2\right)}}
\newcommandx*{\LDAUomega}[2][1=@pkling_false]{\omega\ifthenelse{\equal{#1}{small}}{\bigl(#2\bigr)}{\left(#2\right)}}
\newcommandx*{\LDAUTheta}[2][1=@pkling_false]{\Theta\ifthenelse{\equal{#1}{small}}{\bigl(#2\bigr)}{\left(#2\right)}}

\newcommand*{\N}{\mathbb{N}}

\newcommand*{\greedy}[1]{\ensuremath{\textsc{Greedy}[{#1}]}\xspace}
\newcommand*{\load}{X}

\title{Self-stabilizing Balls \& Bins in Batches\\{\Large The Power of Leaky Bins}}
\date{}

\author[1]{Petra Berenbrink\thanks{petra@cs.sfu.ca.}}
\author[2]{Tom Friedetzky\thanks{tom.friedetzky@dur.ac.uk.}}
\author[1]{Peter Kling\thanks{pkling@sfu.ca. Supported in part by the Pacific Institute for the Mathematical Sciences.}}
\author[1,3]{\\ Frederik Mallmann-Trenn\thanks{mallmann@di.ens.fr}}
\author[4]{Lars Nagel\thanks{nagell@uni-mainz.de. Supported by the German Ministry of Education and Research under Grant 01IH13004.}} 
\author[2]{Chris Wastell\thanks{christopher.wastell@dur.ac.uk. Supported in part by EPSRC.}}

\affil[1]{Simon Fraser University, Burnaby, B.C., V5A 1S6, Canada}
\affil[2]{Durham University, Durham DH1 3LE, U.K.}
\affil[3]{École normale supérieure, 75005 Paris,  France} 
\affil[4]{Johannes Gutenberg-Universität Mainz, 55128 Mainz, Germany}

\begin{document}
\maketitle
\begin{abstract}
A fundamental problem in distributed computing is the distribution of requests to a set of uniform servers without a centralized controller.
Classically, such problems are modelled as static balls into bins processes, where $m$ balls (tasks) are to be distributed to $n$ bins (servers).
In a seminal work, \textcite{ABKU99} proposed the sequential strategy \greedy{d} for $n=m$.
When thrown, a ball queries the load of $d$ random bins and is allocated to a least loaded of these.
\Citeauthor{ABKU99} showed that $d=2$ yields an exponential improvement compared to $d=1$.
\Textcite{BCSV06} extended this to $m\gg n$, showing that the maximal load difference is independent of $m$ for $d=2$ (in contrast to $d=1$).

We propose a new variant of an \emph{infinite} balls into bins process.
Each round an expected number of $\lambda n$ new balls arrive and are distributed (in parallel) to the bins.
Each non-empty bin deletes one of its balls.
This setting models a set of servers processing incoming requests, where clients can query a server's current load but receive no information about parallel requests.
We study the \greedy{d} distribution scheme in this setting and show a strong self-stabilizing property:
For \emph{any} arrival rate $\lambda=\lambda(n)<1$, the system load is time-invariant.
Moreover, for \emph{any} (even super-exponential) round $t$, the maximum system load is (w.h.p.) $\LDAUOmicron[small]{\frac{1}{1-\lambda}\cdot\log\frac{n}{1-\lambda}}$ for $d=1$ and $\LDAUOmicron[small]{\log\frac{n}{1-\lambda}}$ for $d=2$.
In particular, \greedy{2} has an exponentially smaller system load for high arrival rates.
\end{abstract}

\clearpage

\section{Introduction}\label{sec:introduction}
One of the fundamental problems in distributed computing is the distribution of requests, tasks, or data items to a set of uniform servers.
In order to simplify this process and to avoid a single point of failure, it is often advisable to use a simple, randomized strategy instead of a complex, centralized controller to allocate the requests to the servers.
In the most naïve strategy (\emph{1-choice}), each client chooses the server where to send its request uniformly at random.
A more elaborate scheme (\emph{2-choice}) chooses two (or more) servers, queries their current loads, and sends the request to a least loaded of these.
Both approaches are typically modelled as balls-into-bins processes~\cite{Gonnet81,RS98,ABKU99,BCSV06,ACMR98,Stemann96,BCNPP15}, where requests are represented as balls and servers as bins.
While the latter approach leads to considerably better load distributions~\cite{ABKU99,BCSV06}, it loses some of its power in parallel settings, where requests arrive in parallel and cannot take each other into account~\cite{ACMR98,Stemann96}.

We propose and study a new infinite and batchwise balls-into-bins process to model the client-server scenario.
In a round, each server (bin) consumes one of its current tasks (balls).
Afterward, (expectedly) $\lambda n$ tasks arrive and are allocated using a given distribution scheme.
The \emph{arrival rate} $\lambda$ is allowed to be a function of $n$ (e.g., $\lambda=1-1/\poly(n)$).
Standard balls-into-bins results imply that, for high arrival rates, with high probability\footnote{%
	An event $\mathcal{E}$ occurs \emph{with high probability} (w.h.p.) if $\Pr{\mathcal{E}} = 1-n^{-\LDAUOmega{1}}$.
} (w.h.p.) each round there is a bin that receives $\LDAUTheta{\log n}$ balls.

Most other infinite processes limit the total number of concurrent balls in the system by $n$~\cite{ABKU99,BCNPP15} and show a fast recovery.
Since we do not limit the number of balls, our process can, in principle, result in an arbitrarily high system load.
In particular, if starting in a high-load situation (e.g., exponentially many balls), we cannot recover in a polynomial number of steps.
Instead, we adapt the following notion of \emph{self-stabilization}:
The system is positive recurrent (expected return time to a low-load situation is finite) and taking a snapshot of the load situation at an \emph{arbitrary} (even super-exponential large) time yields (w.h.p.) a low maximum load.
Positive recurrence is a standard notion for stability and basically states that the system load is time-invariant.
For irreducible, aperiodic Markov chains it implies the existence of a unique stationary distribution (cf.~Section~\ref{sec:model}).
While this alone does not guarantee a good load in the stationary distribution, together with the snapshot property we can look at an arbitrarily time window of polynomial size (even if it is exponentially far away from the start situation) and give strong load guarantees.
In particular, we give the following bounds on the load in addition to showing positive recurrence:
\begin{description}
\item[1-choice Process:]
	The maximum load at an arbitrary time is (w.h.p.) bounded by $\LDAUOmicron[small]{\frac{1}{1-\lambda}\cdot\log\frac{n}{1-\lambda}}$.
	We also provide a lower bound which is asymptotically tight for $\lambda \leq 1-1/\poly(n)$.
	While this implies that already the simple 1-choice process is self-stabilizing, 
	 the load properties in a \enquote{typical} state are poor: even an arrival rate of only $\lambda=1-1/n$ yields a superlinear maximum load.
\item[2-choice Process:]
	The maximum load at an arbitrary time is (w.h.p.) bounded by $\LDAUOmicron[small]{\log\frac{n}{1-\lambda}}$.
	This allows to maintain an exponentially better system load compared to the 1-choice process; for any $\lambda=1-1/\poly(n)$ the maximum load remains logarithmic.
\end{description}

\subsection{Related Work}\label{sec:relatedwork}
Let us continue with an overview of related work.
We start with classical results for sequential and finite balls-into-bins processes, go over to parallel settings, and give an overview over infinite and batch-based processes similar to ours.
We also briefly mention some results from queuing theory (which is related but studies slightly different quality of service measures and system models).

\paragraph{Sequential Setting.}
There are many strong, well-known results for the classical, sequential balls-into-bins process.
In the sequential setting, $m$ balls are thrown one after another and allocated to $n$ bins.
For $m=n$, the maximum load of any bin is known to be (w.h.p.) $(1+\LDAUomicron{1})\cdot\ln(n)/\ln\ln n$ for the 1-choice process~\cite{Gonnet81, RS98} and $\ln\ln(n)/\ln d+\LDAUTheta{1}$ for the $d$-choice process with $d\geq2$~\cite{ABKU99}.
If $m\geq n\cdot\ln n$, the maximum load increases to $m/n+\LDAUTheta[small]{\sqrt{m\cdot\ln(n)/n}}$~\cite{RS98} and $m/n+\ln\ln(n)/\ln d+\LDAUTheta{1}$~\cite{BCSV06}, respectively.
In particular, note that the number of balls above the average grows with $m$ for $d=1$ but is independent of $m$ for $d\geq2$.
This fundamental difference is known as the \emph{power of two choices}.
A similar (if slightly weaker) result was shown by \textcite{TW14} using a quite elegant proof technique (which we also employ and generalize for our analysis in Section~\ref{sec:two_choice}).
\Textcite{CS97} study adaptive allocation processes where the number of a ball's choices depends on the load of queried bins.
The authors subsequently analyze a scenario that allows reallocations.

\Textcite{BKSS13} adapt the threshold protocol from~\cite{ACMR98} (see below) to a sequential setting and $m\geq n$ bins.
Here, ball $i$ randomly choose a bin until it sees a load smaller than $1+i/n$.
While this is a relatively strong assumption on the balls, this protocol needs only $\LDAUOmicron{m}$ choices in total (allocation time) and achieves an almost optimal maximum load of $\ceil{m/n}+1$.

\paragraph{Parallel Setting.}
Several papers (e.g.~\cite{ACMR98, Stemann96}) investigated parallel settings of multiple-choice games for the case $m=n$.
Here, all $m$ balls have to be allocated in parallel, but balls and bins might employ some (limited) communication.
\Textcite{ACMR98} consider a trade-off between the maximum load and the number of communication rounds $r$ the balls need to decide for a target bin.
Basically, bounds that are close to the classical (sequential) processes can only be achieved if $r$ is close to the maximum load~\cite{ACMR98}.
The authors also give a lower bound on the maximum load if $r$ communication rounds are allowed, and \textcite{Stemann96} provides a matching upper bound via a collision-based protocol.

\paragraph{Infinite Processes.}
In infinite processes, the number of balls to be thrown is not fixed.
Instead, in each of infinitely many rounds, balls are thrown or reallocated while and bins possibly delete old balls.
\Textcite{ABKU99} consider an infinite, sequential process starting with $n$ balls arbitrarily assigned to $n$ bins.
In each round one random ball is reallocated using the $d$-choice process.
For any $t>cn^2\log\log n$, the maximum load at time $t$ is (w.h.p.) $\ln\ln(n)/\ln d+\LDAUOmicron{1}$.

\Textcite{ABS98} consider a system where in each round $m\leq n/9$ balls are allocated.
Bins have a FIFO-queue, and each arriving ball is stored in the queue of two random bins.
After each round, every non-empty bin deletes its frontmost ball (which automatically removes its copy from the second random bin).
It is shown that the expected waiting time is constant and the maximum waiting time is (w.h.p.) $\ln\ln(n)/\ln d+\LDAUOmicron{1}$.
The restriction $m\leq n/9$ is the major drawback of this process.
A differential and experimental study of this process was conducted in~\cite{BCFV00}.
The balls' arrival times are binomially distributed with parameters $n$ and $\lambda=m/n$.
Their results indicate a stable behaviour for $\lambda\leq0.86$.
A similar model was considered by \textcite{Mitzenmacher01}, who considers ball arrivals as a Poisson stream of rate $\lambda n$ for $\lambda<1$.
It is shown that the $2$-choice process reduces the waiting time exponentially compared to the $1$-choice process.

\Textcite{Czumaj98} presents a framework to study the recovery time of discrete-time dynamic allocation processes.
In each round one of $n$ balls is reallocated using the $d$-choice process.
The ball is chosen either by selecting a random bin or by selecting a random ball.
From an arbitrary initial assignment, the system is shown to recover to the maximum load from~\cite{ABKU99} within $\LDAUOmicron{n^2\ln n}$ rounds in the former and $\LDAUOmicron{n\ln n}$ rounds in the latter case.
\Textcite{BCNPP15} consider a similar process with only one random choice per ball, also starting from an arbitrary initial assignment of $n$ balls.
In each round, one ball is chosen from every non-empty bin and reallocated randomly.
The authors define a configuration to be \emph{legitimate} if the maximum load is $\LDAUOmicron{\log n}$.
They show that (w.h.p.) any state recovers in linear time to a legitimate state and maintain such a state for $\poly(n)$ rounds.

\paragraph{Batch-Processes.}
Batch-based processes allocate $m$ balls to $n$ bins in batches of (usually) $n$ balls each, where each batch is allocated in parallel.
They lie between (pure) parallel and sequential processes.
For $m=\tau\cdot n$, \textcite{Stemann96} investigates a scenario with $n$ players each having $m/n$ balls.
To allocate a ball, every player independently chooses two bins and allocates copies of the ball to both of them.
Every bin has two queues (one for first copies, one for second copies) and processes one ball from each queue per round.
When a ball is processed, its copy is removed from the system and the player is allowed to initiate the allocation of the next ball.
If $\tau=\ln n$, all balls are processed in $\LDAUOmicron{\ln n}$ rounds and the waiting time is (w.h.p.) $\LDAUOmicron{\ln\ln n}$.
\Textcite{BCEFN12} study the $d$-choice process in a scenario where $m$ balls are allocated to $n$ bins in batches of size $n$ each. 
The authors show that the load of every bin is (w.h.p.) $m/n\pm\LDAUOmicron{\log n}$.
As noted in Lemma~\ref{lem:two_choice:smoothness}, our analysis can be used to derive the same result by easier means.

\paragraph{Queuing Processes.}
Batch arrival processes have also been considered in the context of queuing systems.
A key motivation for such models stems from the asynchronous transfer mode (ATM) in telecommunication systems.
Tasks arrive in batches and are stored in a FIFO queue.
Several papers~\cite{Sohraby92,Kamal96,Kim12,Alfa03} consider scenarios where the number of arriving tasks is determined by a finite state Markov chain.
Results study steady state properties of the system to determine properties of interest (e.g., waiting times or queue lengths). 
\Textcite{Sohraby92} use spectral techniques to study a multi-server scenario with an infinite queue.
\Textcite{Alfa03} considers a discrete-time process for $n$ identical servers and tasks with constant service time $s\geq1$.
To ensure a stable system, the arrival rate $\lambda$ is assumed to be $\leq n/s$ and tasks are assigned cyclical, allowing to study an arbitrary server (instead of the complete system).
\Textcite{Kamal96} and \Textcite{Kim12} study a system with a finite capacity.
Tasks arriving when the buffer is full are lost.
The authors study the steady state probability and give empirical results to show the decay of waiting times as $n$ increases.

\subsection{Model \& Preliminaries}\label{sec:model}
We model our load balancing problem as an infinite, parallel balls-into-bins processes.
Time is divided into discrete, synchronous rounds.
There are $n$ bins and $n$ generators, and the initial system is assumed to be empty.
At the start of each round, every non-empty bins deletes one ball.
Afterward, every generator generates a ball with a probability of $\lambda=\lambda(n)\in\intcc{0,1}$ (the \emph{arrival rate}).
This generation scheme allows us to consider arrival rates that are arbitrarily close to one (like $1-1/\poly(n)$).
Generated balls are distributed in the system using a distribution process.
In this paper we analyze two specific distribution processes:
\begin{enumerate*}
\item The $1$-choice process \greedy{1} assigns every ball to a randomly chosen bin.
\item The $2$-choice process \greedy{2} assigns every ball to a least loaded among two randomly chosen bins.
\end{enumerate*}

\paragraph{Notation.}
The random variable $\load_i(t)$ denotes the load (number of balls) of the $i$-th fullest bin at the end of round $t$.
Thus, the load situation (configuration) after round $t$ can be described by the load vector $\bm{\load}(t)=(\load_i(t))_{i\in\intcc{n}}\in\N^n$.
We define $\varnothing(t)\coloneqq\frac1n \sum_{i=1}^n\load_i(t)$ as the average load at the end of round $t$.
The value $\nu(t)$ denotes the fraction of non-empty bins after round $t$ and $\eta(t)\coloneqq1-\nu(t)$ the fraction of empty bins after round $t$.
It will be useful to define $\mathbbm{1}_i(t)\coloneqq\min\bigl(1,\load_i(t)\bigr)$ and $\eta_i(t)\coloneqq \mathbbm{1}_i(t)-\nu(t)$ (which equals $\eta(t)$ if $i$ is a non-empty bin and $-\nu(t)$ otherwise).

\paragraph{Markov Chain Preliminaries.}
The evolution of the load vector over time can be interpreted as a Markov chain, since $\bm{\load}(t)$ depends only on $\bm{\load}(t-1)$ and the random choices during round $t$.
We refer to this Markov chain as $\bm{\load}$.
Note that $\bm{\load}$ is time-homogeneous (transition probabilities are time-independent), irreducible (every state is reachable from every other state), and aperiodic (path lengths have no period; in fact, our chain is lazy).
Recall that such a Markov chain is positive recurrent (or ergodic) if the probability to return to the start state is $1$ and the expected return time is finite.
In particular, this implies the existence of a unique stationary distribution.
Positive recurrence is a standard formalization of the intuitive concept of stability.
See~\cite{Levin:2008} for an excellent introduction into Markov chains and the involved terminology.


\section{The 1-Choice Process}\label{sec:one_choice}
We present two main results for the $1$-choice process:
Theorem~\ref{thm:one_choice:stability} states the stability of the system under the 1-choice process for an arbitrary $\lambda$, using the standard notion of positive recurrence (cf.~Section~\ref{sec:introduction}).
In particular, this implies the existence of a stationary distribution for the 1-choice process.
Theorem~\ref{thm:one_choice:maxload} strengthens this by giving a high probability bound on the maximum load for an \emph{arbitrary} round $t\in\N$.
Together, both results imply that the 1-choice process is self-stabilizing.
\begin{theorem}[Stability]\label{thm:one_choice:stability}
Let $\lambda=\lambda(n)<1$.
The Markov chain $\bm{\load}$ of the 1-choice process is positive recurrent.
\end{theorem}
\begin{theorem}[Maximum Load]\label{thm:one_choice:maxload}
Let $\lambda=\lambda(n)<1$.
Fix an arbitrary round $t$ of the 1-choice process.
The maximum load of all bins is (w.h.p.) bounded by $\LDAUOmicron[small]{\frac{1}{1-\lambda}\cdot\log\frac{n}{1-\lambda}}$.
\end{theorem}

Note that for high arrival rates of the form $\lambda(n)=1-\varepsilon(n)$, the bound given in Theorem~\ref{thm:one_choice:maxload} is inversely proportional to $\varepsilon(n)$.
For example, for $\varepsilon(n)=1/n$ the maximal load is $\LDAUOmicron{n\log n}$.
Theorem~\ref{thm:one_choice:lowerbound} shows that this dependence is unavoidable: the bound given in Theorem~\ref{thm:one_choice:maxload} is tight for large values of $\lambda$.
In Section~\ref{sec:two_choice}, we will see that the 2-choice process features an exponentially better behaviour for large $\lambda$.
\begin{theorem}\label{thm:one_choice:lowerbound}
Let $\lambda=\lambda(n)\geq0.5$ and define $t\coloneqq9\lambda\log\left(n\right)/(64(1-\lambda)^2)$.
With probability $1-\LDAUomicron{1}$ there is a bin $i$ in step $t$ with load $\LDAUOmega[small]{\frac{1}{1-\lambda}\cdot\log n}$.
\end{theorem}
The proofs of these results can be found in the following subsections. We first prove a bound on the maximum load (Theorem~\ref{thm:one_choice:maxload}), afterwards we prove stability of the system (Theorem~\ref{thm:one_choice:stability}), and finally we prove the lower bound (Theorem~\ref{thm:one_choice:lowerbound}).

\subsection{Maximum Load – Proof of Theorem~\ref{thm:one_choice:maxload}}
\begin{proof}[Proof of Theorem~\ref{thm:one_choice:maxload} (Maximum Load)]
We prove Theorem~\ref{thm:one_choice:maxload} using a (slightly simplified) drift theorem from~\textcite{H82} (cf.~Theorem~\ref{thm:Hajek} in Appendix~\ref{app:auxiliary}).
Remember that, as mentioned in Section~\ref{sec:model}, our process is a Markov chain, such that we need to condition only on the previous state (instead of the full filtration from Theorem~\eqref{thm:Hajek}).
Our goal is to bound the load of a fixed bin $i$ at time $t$ using Theorem~\ref{thm:Hajek} and, subsequently, to use this with a union bound to bound the maximum load over all bins.
To apply Theorem~\ref{thm:Hajek}, we have to prove that the maximum load difference of bin $i$ between two rounds is is exponentially bounded (Majorization) and that, given a high enough load, the system tends to loose load (Negative Bias).
We start with the majorization.
The load difference $\abs{\load_i(t+1)-\load_i(t)}$ is bounded by $\max(1,B_i(t))\leq1+B_i(t)$, where $B_i(t)$ is the number of tokens resource $i$ receives during round $t+1$.
In particular, we have $(\abs{\load_i(t+1)-\load_i(t)}\,|\,\bm{\load}(t))\prec1+B_i(t)$.
Note that $B_i(t)$ is binomially distributed with parameters $n$ and $\lambda/n$ (each of the $n$ balls has probability of $\lambda\cdot1/n$ to end up in $i$).
Using standard inequalities we bound
\begin{equation}
\Pr{B_i(t)=k}\leq\binom{n}{k}\cdot\left(\frac{\lambda}{n}\right)^k\leq\left(\frac{e\cdot n}{k}\right)^k\cdot\left(\frac{1}{n}\right)^k=\frac{e^k}{k^k}
\end{equation}
and calculate
\begin{equation}
     \Ex{e^{B_i(t)+1}}
=    e\cdot\sum_{k=0}^ne^k\cdot\frac{e^k}{k^k}
\leq e\cdot\sum_{k=0}^{\ceil{e^3-1}}\frac{e^{2k}}{k^k}+e\cdot\sum_{k=e^3}^{\infty}\frac{e^{2k}}{k^k}
\leq \LDAUTheta{1}+\sum_{k=1}^{\infty}e^{-k}
=    \LDAUTheta{1}
.
\end{equation}
This shows that the Majorization condition from Theorem~\ref{thm:Hajek} holds (with $\lambda'=1$ and $D=\LDAUTheta{1}$).
To see that the Negative Bias condition is also given, note that if bin $i$ has non-zero load, it is guaranteed to delete one ball and receives in expectation $n\cdot\lambda/n=\lambda$ balls.
We get $\Ex{\load_i(t+1)-\load_i(t)|\load_i(t)>0}\leq\lambda-1<0$, establishing the Negative Bias condition (with $\varepsilon_0=1-\lambda$).
We finally can apply Theorem~\ref{thm:Hajek} with $\eta\coloneqq\min(1,(1-\lambda)/2D,1/(2-2\lambda))=(1-\lambda)/(2D)$ and get for $b\geq0$
\begin{equation}\label{eq:one_choice:maxload:eq1}
     \Pr{\load_i(t)\geq b}
\leq e^{-b\cdot\eta}+\frac{2D}{\eta\cdot(1-\lambda)}\cdot e^{\eta\cdot(1-b)}
\leq \frac{2\cdot(2D)^2}{(1-\lambda)^2}\cdot e^{\frac{(1-\lambda)\cdot(1-b)}{2D}}
=    \frac{c}{(1-\lambda)^2}\cdot e^{-\frac{b\cdot(1-\lambda)}{c}}
,
\end{equation}
where $c$ denotes a suitable constant.
Applying a union bound to all $n$ bins and choosing $b\coloneqq\frac{c}{1-\lambda}\cdot\ln\bigl(\frac{c\cdot n^{a+1}}{(1-\lambda)^2}\bigr)$ yields $\Pr{\max_{i\in\intcc{n}}\load_i(t)\geq b}\leq n^{-a}$.
The theorem's statement now follows from
\begin{equation}
     b
=    \frac{c}{1-\lambda}\cdot\ln\left(\frac{c\cdot n^{a+1}}{(1-\lambda)^2}\right)
\leq \frac{c\cdot(a+1)+1}{1-\lambda}\cdot\ln\left(\frac{n}{1-\lambda}\right)
=    \LDAUOmicron{\frac{1}{1-\lambda}\cdot\ln\left(\frac{n}{1-\lambda}\right)}
.
\end{equation}
\end{proof}

\subsection{Stability – Proof of Theorem~\ref{thm:one_choice:stability} }

In the following, we provide an auxiliary lemma that will prove useful to derive the stability of the 1-choice process.
\begin{lemma}\label{lem:one_choice:expected_bin_load}
Let $\lambda=\lambda(n)<1$.
Fix an arbitrary round $t$ of the 1-choice process and a bin $i$.
There is a constant $c>0$ such that the expected load of bin $i$ is bounded by $\frac{6c}{1-\lambda}\cdot\ln\bigl(\frac{e\cdot c}{1-\lambda}\bigr)$.
\end{lemma}
\begin{proof}
To get a bound on the expected load of bin $i$, note that the probability in Equation~\eqref{eq:one_choice:maxload:eq1} (see proof of Theorem~\ref{thm:one_choice:maxload}) is $1$ for $b\leq\gamma\coloneqq\frac{c}{1-\lambda}\cdot\ln\bigl(\frac{e\cdot c}{(1-\lambda)^2}\bigr)$.

Considering time windows of $\gamma$ rounds each, we calculate
\begin{equation}
\begin{aligned}
      \Ex{\load_i(t)}
&\leq \sum_{b=1}^{\gamma}b\cdot\Pr{\load_i(t)=b}+\sum_{k=1}^\infty \sum_{b=k\cdot \gamma}^{(k+1)\gamma}b\cdot\Pr{\load_i(t)=b}\\
&\leq \gamma+\sum_{k=1}^{\infty}(k+1)\cdot\gamma\cdot\Pr{\load_i(t)\geq k\cdot\gamma}
 \leq \gamma+\sum_{k=1}^{\infty}(k+1)\cdot\gamma\cdot e^{-k}\\
&\leq 3\gamma
 \leq \frac{6c}{1-\lambda}\cdot\ln\left(\frac{e\cdot c}{1-\lambda}\right)
.
\end{aligned}
\end{equation}
This finishes the proof.
\end{proof}

\begin{proof}[Proof of Theorem~\ref{thm:one_choice:stability} (Stability)]
We prove Theorem~\ref{thm:one_choice:stability} using a result from~\textcite{FMM95} (cf.~Theorem~\ref{thm:foster} in Appendix~\ref{app:auxiliary}).
Note that $\bm{\load}$ is a time-homogenous irreducible Markov chain with a countable state space.
For a configuration $\bm{x}$ we define the auxiliary potential $\Psi(\bm{x})\coloneqq\sum_{i=1}^nx_i$ as the total system load of configuration $\bm{x}$.
Consider the (finite) set $C\coloneqq\set{\bm{x}|\Psi(\bm{x})\leq n^4/(1-\lambda)^2}$ of all configurations with not too much load.
To prove positive recurrence, it remains to show that Condition~\ref{thm:foster:a} (expected potential drop if not in a high-load configuration) and Condition~\ref{thm:foster:b} (finite potential) of Theorem~\ref{thm:foster} hold.
In the following, let $\Delta\coloneqq\frac{n^3}{(1-\lambda)^2}$.

Let us start with Condition~\ref{thm:foster:a}.
So fix a round $t$ and let $\bm{x}=\bm{\load}(t)\not\in C$.
By definition of $C$, we have $\Psi(\bm{x})>n^4/(1-\lambda)^3$, such that there is at least one bin $i$ with load $x_i\geq\Psi(\bm{x})/n>n^3/(1-\lambda)^2$.
In particular, note that $x_i\geq\Delta$, such that during each of the next $\Delta$ rounds exactly one ball is deleted.
On the other hand, bin $i$ receives in expectation $\Delta\cdot\lambda n\cdot\frac{1}{n}=\lambda\Delta$ balls during the next $\Delta$ rounds.
We get $\Ex{\load_i(t+\Delta)-x_i|\bm{\load}(t)=\bm{x}}=\lambda\Delta-\Delta=-(1-\lambda)\cdot\Delta$.
For any bin $j\neq i$, we assume pessimistically that no ball is deleted.
Note that the expected load increase of each of these bins can be majorized by the load increase in an empty system running for $\Delta$ rounds.
Thus, we can use Lemma~\ref{lem:one_choice:expected_bin_load} to bound the expected load increase in each of these bins by $\frac{6c}{1-\lambda}\cdot\ln\bigl(\frac{e\cdot c}{1-\lambda}\bigr)\leq\frac{6e\cdot c^2}{(1-\lambda)^2}\leq\Delta/n^2$.
We get
\begin{equation*}
     \Ex{\Psi(\bm{X}(t+\Delta))|\bm{X}(t)=\bm{x}}
\leq -(1-\lambda)\cdot\Delta+(n-1)\cdot\frac{\Delta}{n^2}
=    -\Delta\cdot\left(1-\lambda-\frac{1}{n}\right)
\leq -\Delta\cdot\frac{1-\lambda}{2}
.
\end{equation*}
This proves Condition~\ref{thm:foster:a} of Theorem~\ref{thm:foster}.
For Condition~\ref{thm:foster:b}, assume $\bm{x}=\bm{\load}(t)\in C$.
We bounds the system load after $\Delta$ rounds trivially by
\begin{equation}
     \Ex{\Psi(\bm{\load}({t+\Delta}))|\bm{\load}(t)=\bm{x}}
\leq \Psi(\bm{x})+\Delta\cdot n
\leq \frac{n^4}{(1-\lambda)^2}+\Delta\cdot n
<    \infty
\end{equation}
(note that the finiteness in Theorem~\ref{thm:foster} is with respect to time, not $n$).
This finishes the proof.
\end{proof}

\subsection{Lower Bound on Maximum Load – Proof of Theorem~\ref{thm:one_choice:lowerbound} }

\begin{proof}[Proof of Theorem~\ref{thm:one_choice:lowerbound} (Lower Bound)]
To show this result we will use the bound of Theorem~\ref{kuchen} which lower bounds the the maximum number 
of balls a bin receives when $m$ balls are allocated into $m$ bins.
The idea of the proof is as follows.
We assume that we start at an empty system and apply Theorem~\ref{kuchen} on $m=\lambda t n$ many balls.
The theorem says that one of the bins is likely to get much more than $\lambda t$ many balls, which allows us 
to show that the load of this bin is large, even if the bin was able to delete a ball during each of the $t$ 
observed time steps.

Let $m(t')$ the the number of balls allocated during the first $t'$ steps and let $b_u(t')$ the number of these balls that are allocated to bin $u$.
Set $t=9\lambda \log\left(n\right)/(64(1-\lambda)^2)$, assume $\lambda>0.5$ and assume
$\lambda n t\le n\cdot (\log n)^c$ for a constant $c$.
Since the expected number of balls is
$\lambda n t\ge  n\log n$ we can use Chernoff bounds to show that w.h.p. at least $(1-\epsilon) \cdot m(t)$ balls are generated for very small $\epsilon$.

Then $$\Ex{m(t')}= \tfrac{9\lambda \log n}{64(1-\lambda)^2}\cdot \lambda n.$$ Using Chernoff's inequality we can show that w.h.p. $m(t)\ge (1-\epsilon) \cdot \frac{9\lambda \log n}{64(1-\lambda)^2}\cdot \lambda n$
for an arbitrary small constant $\epsilon$.
By Theorem~\ref{kuchen} (Case 3) with $\alpha=\sqrt{8/9}$ we get (w.h.p.)
\begin{align}\label{apfelkuchen}
b_u(t)\geq  (1-\epsilon) \cdot \frac{9\lambda \log n}{64(1-\lambda)^2}\cdot \lambda + 
\sqrt{(1-\epsilon) \cdot \frac{16\cdot 9\lambda (\log n)^2}{9 \cdot 64(1-\lambda)^2}\cdot \lambda}
.
\end{align}
We derive 
\begin{equation}\label{capful}
\begin{aligned}
\load_u(t) 
&\geq  (1-\epsilon)\cdot \frac{9}{64}\frac{\lambda^2\log n}{(1-\lambda)^2} + 
\sqrt{(1-\epsilon)\cdot \frac{16\cdot 9\lambda^2\cdot (\log n)^2}{9\cdot 64 (1-\lambda)^2}}   -   \frac{9\lambda \log n}{64(1-\lambda)^2}  \notag\\
&=  \lambda\cdot (1-\epsilon)\cdot \frac{9}{64}\frac{\lambda\log n}{(1-\lambda)^2} + 
\sqrt{\frac{(1-\epsilon)}{4 }} \cdot\frac{\lambda \log n}{(1-\lambda)}   -   \frac{9\lambda \log n}{64(1-\lambda)^2} 
\notag\\
&= \sqrt{\frac{(1-\epsilon)}{4 }} \cdot\frac{\lambda \log n}{(1-\lambda)} +  (\lambda\cdot (1-\epsilon)-1)\cdot
\frac{9 \lambda\log n}{64(1-\lambda)^2}  \notag\\
&\ge\sqrt{\frac{(1-\epsilon)}{4 }} \cdot\frac{\lambda \log n}{(1-\lambda)} +   ((1-\epsilon')-1)\cdot
\frac{9 \lambda\log n}{64(1-\lambda)^2}  \notag\\
&\ge   \sqrt{\frac{(1-\epsilon)}{4 }} \cdot\frac{\lambda \log n}{(1-\lambda)} 
- \epsilon' \cdot \frac{9}{64}\frac{\lambda\log n}{(1-\lambda)}  \notag\\
&=\Omega\left(\frac{\lambda \log n}{1-\lambda}\right)
.
\end{aligned}
\end{equation}

\end{proof}


\section{The 2-Choice Process}\label{sec:two_choice}
We continue with the study of the 2-choice process.
Here, new balls are distributed according to \greedy{2} (cf.~description in Section~\ref{sec:model}).
Our main results are the following theorems, which are equivalents to the corresponding theorems for the 1-choice process.
\begin{theorem}[Stability]\label{thm:two_choice:stability}
Let $\lambda=\lambda(n)\in\intco{1/4,1}$. 
The Markov chain $\bm{\load}$ of the 2-choice process is positive recurrent.
\end{theorem}
\begin{theorem}[Maximum Load]\label{thm:two_choice:maxload}
Let $\lambda=\lambda(n)\in\intco{1/4,1}$.
Fix an arbitrary round $t$ of the 2-choice process.
The maximum load of all bins is (w.h.p.) bounded by $\LDAUOmicron[small]{\log\frac{n}{1-\lambda}}$.
\end{theorem}
Note that Theorem~\ref{thm:two_choice:maxload} implies a much better behaved system than we saw in Theorem~\ref{thm:one_choice:maxload} for the 1-choice process.
In particular, it allows for an exponentially higher arrival rate:
for $\lambda(n)=1-1/\poly(n)$ the 2-choice process maintains a maximal load of $\LDAUOmicron{\log n}$.
In contrast, for the same arrival rate the 1-choice process results in a system with maximal load $\LDAUOmega{\poly(n)}$.

Our analysis of the 2-choice process relies to a large part on a good bound on the \emph{smoothness} (the maximum load difference between any two bins).
This is stated in the following lemma.
This result is of independent interest, showing that even if the arrival rate is $1-e^{-n}$, where we get a polynomial system load, the maximum load difference is still logarithmic.
\begin{lemma}[Smoothness]\label{lem:two_choice:smoothnessWHP}
Let $\lambda=\lambda(n)\in\intcc{1/4,1}$.
Fix an arbitrary round $t$ of the 2-choice process.
The load difference of all bins is (w.h.p.) bounded by $\LDAUOmicron{\ln n}$.
\end{lemma}

\paragraph{Analysis Overview.}
To prove these results, we combine three different potential functions:
For a configuration $\bm{x}$ with average load $\varnothing$ and for a suitable constant $\alpha$ (to be fixed later), we define
\begin{equation}\label{eqn:two_choice:potentials}
\begin{aligned}
\Phi(\bm{x})   &\coloneqq \sum_{i\in\intcc{n}}e^{\alpha\cdot(x_i-\varnothing)}+\sum_{i\in\intcc{n}}e^{\alpha\cdot(\varnothing-x_i)},
\qquad\qquad
\Psi(\bm{x})    \coloneqq \sum_{i\in\intcc{n}}x_i,\quad\text{and}\\
\Gamma(\bm{x}) &\coloneqq \Phi{(\bm{x})} + \tfrac{n}{1-\lambda} \cdot\Psi(\bm{x})
.
\end{aligned}
\end{equation}
The potential $\Phi$ measures the \emph{smoothness} (basically the maximum load difference to the average) of a configuration and is used to prove Lemma~\ref{lem:two_choice:smoothnessWHP} (Section~\ref{sec:two_choice:bound_smoothness}).
The proof is based on the observation that whenever the load of a bin is far from the average load, it decreases in expectation.
The potential $\Psi$ measures the \emph{total load} of a configuration and is used, in combination with our results on the smoothness, to prove Theorem~\ref{thm:two_choice:maxload} (Section~\ref{sec:two_choice:maxload}).
The potential $\Gamma$ entangles the smoothness and total load, allowing us to prove Theorem~\ref{thm:two_choice:stability} (Section~\ref{sec:two_choice:stability}).
The proof is based on the fact that whenever $\Gamma$ is large (i.e., the configuration is not smooth or it has a huge total load) it decreases in expectation.

Before we continue with our analysis, let us make a simple but useful observation concerning the smoothness:
For any configuration $\bm{x}$ and value $b\geq0$, the inequality $\Phi(\bm{x})\leq e^{\alpha\cdot b}$ implies (by definition of $\Phi$) $\max_i\abs{x_i-\varnothing}\leq b$.
That is, the load difference of any bin to the average is at most $b$ and, thus, the load difference between any two bins is at most $2b$.
We capture this in the following observation.
\begin{observation}\label{obs:two_choice:avg_distance}
Let $b\geq0$ and consider a configuration $\bm{x}$ with average load $\varnothing$.
If $\Phi(\bm{x})\leq e^{\alpha\cdot b}$, then $\abs{x_i-\varnothing}\leq b$ for all $i\in\intcc{n}$.
In particular, $\max_i(x_i)-\min_i(x_i)\leq2b$.
\end{observation}

\subsection{Bounding the Smoothness}\label{sec:two_choice:bound_smoothness}
The goal of this section is to prove Lemma~\ref{lem:two_choice:smoothnessWHP}.
To do so, we show the following bound on the expected smoothness (potential $\Phi)$ at an arbitrary time $t$:
\begin{lemma}\label{lem:two_choice:smoothness}
Let $\lambda\in\intcc{1/4,1}$.
Fix an arbitrary round $t$ of the 2-choice process.
There is a constant $\varepsilon>0$ such that\footnote{%
	For $\Phi$, the condition $\lambda\geq1/4$ can be substituted with $\lambda=\LDAUOmega{1}$ and only minor changes in the analysis.
	Moreover, the analysis can be easily adapted for a process that (deterministically) throws $\lambda\cdot n$ balls in each round, even for $\lambda>1$ as long as it is a constant.
	Finally, one can easily adapt the analysis to cover the process without deletions by setting $\eta_i(t)=0$ (see Observation~\ref{obs:two_choice:oneround_potdrop}).
	Using Markov's inequality, this yields the same result as~\cite{BCEFN12} using a simpler analysis.
}
\begin{equation}
\Ex{\Phi(\bm{\load}(t))}\leq\frac{n}{\varepsilon}
.
\end{equation}
\end{lemma}
Note that Lemma~\ref{lem:two_choice:smoothness} together with Observation~\ref{obs:two_choice:avg_distance} immediately implies Lemma~\ref{lem:two_choice:smoothnessWHP} by a simple application of Markov's inequality to bound the probability that $\Phi(\bm{\load}(t))\geq n^2/\varepsilon$.

Our proof of Lemma~\ref{lem:two_choice:smoothness} follows the lines of~\cite{Peres:2010:CPW:1873601.1873732,TW14}, who used the same potential function to analyze variants of the sequential $d$-choice process without deletions.
While the basic idea of showing a relative drop when the potential is high combined with a bounded absolute increase in the general case is the same, our analysis turns out much more involved.
In particular, not only do we have to deal with deletions and throwing balls in batches but the size of each batch is also a random variable.
Once Lemma~\ref{lem:two_choice:smoothness} is proven, Lemma~\ref{lem:two_choice:smoothnessWHP} emerges by combining Observation~\ref{obs:two_choice:avg_distance}, Lemma~\ref{lem:two_choice:smoothness}, and Markov's inequality as follows:
\begin{equation}
\Pr{\max_i\load_i(t)-\min_i\load_i(t)\geq\frac{4}{\alpha}\cdot\ln\left(\frac{n}{\varepsilon}\right)}\leq\Pr{\Phi(\bm{\load}(t))\geq\frac{n^2}{\varepsilon^2}}\leq\frac{\varepsilon}{n}
.
\end{equation}
It remains to prove Lemma~\ref{lem:two_choice:smoothness}.
Remember the definition of $\Phi(\bm{x})$ from Equation~\eqref{eqn:two_choice:potentials}.
We split the potential in two parts $\Phi(\bm{x})\coloneqq\Phi_+(\bm{x})+\Phi_-(\bm{x})$.
Here, $\Phi_+(\bm{x})\coloneqq\sum_ie^{\alpha\cdot(x_i-\varnothing))}$ denotes the \emph{upper potential} of $\bm{x}$ and $\Phi_-(\bm{x})\coloneqq\sum_ie^{\alpha\cdot(\varnothing-x_i))}$ denotes the \emph{lower potential} of $\bm{x}$.
For a fixed bin $i$, we use $\Phi_{i,+}(\bm{x})\coloneqq e^{\alpha\cdot(x_i-\varnothing)}$ and $\Phi_{i,-}(\bm{x})\coloneqq e^{\alpha\cdot(\varnothing-x_i)}$ to denote $i$'s contribution to the upper and lower potential, respectively.
When we consider the effect of a fixed round $t+1$, we will sometimes omit the time parameter and use prime notation to denote the value of a parameter at the end of round $t+1$.
For example, we write $\load_i$ and $\load'_i$ for the load of bin $i$ at the beginning and at the end of round $t+1$, respectively.

We start with two simple but useful identities regarding the potential drop $\Delta_{i,+}(t+1)$ (and $\Delta_{i,-}(t+1)$) due to a fixed bin $i$ during round $t+1$.
\begin{observation}\label{obs:two_choice:oneround_potdrop}
Fix a bin $i$, let $K$ denote the number of balls that are placed during round $t+1$ and let $k\leq K$ be the number of these balls that fall into bin $i$.
Then
\begin{enumerate}
\item $\Delta_{i,+}(t+1)=\Phi_{i,+}(\bm{X}(t))\cdot\left(e^{\alpha\cdot(k-\eta_i(t)-K/n)}-1\right)$ and
\item $\Delta_{i,-}(t+1)=\Phi_{i,-}(\bm{X}(t))\cdot\left(e^{-\alpha\cdot(k-\eta_i(t)-K/n)}-1\right)$.
\end{enumerate}
\end{observation}

We now derive the main technical lemma that states general bounds on the expected upper and lower potential change during a single round.
This will be used to derive bounds on the potential change in different situations.
For this, let $p_i\coloneqq(\frac{i}{n})^2-(\frac{i-1}{n})^2=\frac{2i-1}{n^2}$ (the probability that a ball thrown with \greedy{2} falls into the $i$-th fullest bin).
We also define $\hat{\alpha}\coloneqq e^{\alpha}-1$ and $\check{\alpha}\coloneqq1-e^{-\alpha}$.
Note that $\hat{\alpha}\in\intoo{\alpha,\alpha+\alpha^2}$ and $\check{\alpha}\in\intoo{\alpha-\alpha^2,\alpha}$ for $\alpha \in\intoo{0, 1.7}$.
This follows easily from the Taylor approximation $e^x\leq1+x+x^2$, which holds for any $x\in\intoc{-\infty,1.7}$ (we will use this approximation several times in the analysis).
Finally, let $\hat{\delta}_i\coloneqq\lambda n\cdot(\sfrac{1}{n}\cdot\check{1}-p_i\cdot\sfrac{\hat{\alpha}}{\alpha})$ and $\check{\delta}_i\coloneqq\lambda n\cdot(\sfrac{1}{n}\cdot\hat{1}-p_i\cdot\sfrac{\check{\alpha}}{\alpha})$, where $\check{1}\coloneqq1-\alpha/n<1<\hat{1}\coloneqq1+\alpha/n$.
These $\hat{\delta}_i$ and $\check{\delta}_i$ values can be thought of as upper/lower bounds on the expected difference in the number of balls that fall into bin $i$ under the 1-choice and 2-choice process, respectively (note that $\hat{1}$, $\check{1}$, $\hat{\alpha}/\alpha$, and $\check{\alpha}/\alpha$ are all constants close to 1).
\begin{lemma}\label{lem:pot_change}
Consider a bin $i$ after round $t$ and a constant $\alpha\leq1$.
\begin{enumerate}
\item For the expected change of $i$'s upper potential during round $t+1$ we have
	\begin{equation}
	\frac{\Ex{\Delta_{i,+}(t+1)|\bm{\load}(t)}}{\Phi_{i,+}(\bm{\load}(t))}\leq-\alpha\cdot\left(\eta_i+\hat{\delta}_i\right)+\alpha^2\cdot\left(\eta_i+\hat{\delta}_i\right)^2
	.
	\end{equation}
\item For the expected change of $i$'s lower potential during round $t+1$ we have
	\begin{equation}
	\frac{\Ex{\Delta_{i,-}(t+1)|\bm{\load}(t)}}{\Phi_{i,-}(\bm{\load}(t))}\leq\alpha\cdot\left(\eta_i+\check{\delta}_i\right)+\alpha^2\cdot\left(\eta_i+\check{\delta}_i\right)^2
	.
	\end{equation}
\end{enumerate}
\end{lemma}
\begin{proof}
For the first statement, we use Observation~\ref{obs:two_choice:oneround_potdrop} to calculate $\Ex{\Delta_{i,+}(t)|\bm{\load}}/\Phi_{i,+}$
\begin{align*}
{}={}& \sum_{K=0}^n\sum_{k=0}^K\binom{n}{K}\binom{K}{k}\cdot(p_i\lambda)^k\cdot\bigl((1-p_i)\lambda\bigr)^{K-k}\cdot(1-\lambda)^{n-K}\cdot\left(e^{\alpha\cdot(k-\eta_i-K/n)}-1\right)\\
{}={}& \sum_{K=0}^n\binom{n}{K}(1-\lambda)^{n-K}\cdot\lambda^K\sum_{k=0}^K\binom{K}{k}\cdot p_i^k\cdot(1-p_i)^{K-k}\cdot\left(e^{\alpha\cdot(k-\eta_i-K/n)}-1\right)\\
{}={}& \sum_{K=0}^n\binom{n}{K}(1-\lambda)^{n-K}\cdot\lambda^K\cdot\left(e^{-\alpha(\eta_i+K/n)}\sum_{k=0}^K\binom{K}{k}\cdot(e^{\alpha}\cdot p_i)^k\cdot(1-p_i)^{K-k}-1\right)\\
{}={}& \sum_{K=0}^n\binom{n}{K}(1-\lambda)^{n-K}\cdot\lambda^K\cdot\left(e^{-\alpha(\eta_i+K/n)}\cdot\left(1+\hat{\alpha}\cdot p_i\right)^K-1\right)
,\\
\intertext{%
	where we first apply the law of total expectation together with Observation~\ref{obs:two_choice:oneround_potdrop} and, afterward, twice the binomial theorem.
	Continuing the calculation using the aforementioned Taylor approximation $e^x\leq1+x+x^2$ (which holds for any $x\in\intoc{-\infty,1.7}$), and the definition of $\hat{\delta}_i$ yields
}
{}={}& e^{-\alpha\eta_i}\cdot\bigl(1-\lambda+\lambda e^{-\alpha/n}\cdot(1+\hat{\alpha}\cdot p_i)\bigr)^n-1\leq e^{-\alpha\eta_i}\cdot\bigl(1-\lambda(1-e^{-\alpha/n})+\lambda\cdot\hat{\alpha}\cdot p_i\bigr)^n-1\\
{}\leq{}& e^{-\alpha\eta_i}\cdot\left(1-\frac{\lambda\cdot\alpha}{n}\cdot(1-\alpha/n)+\lambda\cdot\hat{\alpha}\cdot p_i\right)^n-1\leq e^{-\alpha\eta_i}\cdot\left(1-\frac{\alpha}{n}\cdot\hat{\delta}_i\right)^n-1
\leq      e^{-\alpha\cdot\left(\eta_i+\hat{\delta}_i\right)}-1
.
\end{align*}
Now, the claim follows by another application of the Taylor approximation.
The second statement follows similarly.
\end{proof}
Using Lemma~\ref{lem:pot_change}, we derive different bounds on the potential drop that will be used in the various situations.
The proofs for the following statements can all be found in Appendix~\ref{app:missing:two_choice}.

We start with a result that will be used when the potential is relatively high.
\begin{lemma}\label{lem:two_choice_potdrop_largepot}
Consider a round $t$ and a constant $\alpha\leq\ln(10/9)$ ($<1/8$).
Let $R\in\set{+,-}$ and $\lambda\in\intcc{1/4,1}$.
For the expected upper and lower potential drop during round $t+1$ we have
\begin{equation}
\Ex{\Delta_R(t+1)|\bm{\load}(t)}<2\alpha\lambda\cdot\Phi_R(\bm{\load}(t))
.
\end{equation}
\end{lemma}

The next lemma derives a bound that is used to bound the upper potential change in reasonably balanced configurations.
\begin{lemma}\label{lem:bal_cfg:bound_upper_pot}
Consider a round $t$ and the constants $\varepsilon$ (from Claim~\ref{clm:bulk_leftright}) and $\alpha\leq\min(\ln(10/9),\varepsilon/4)$.
Let $\lambda\in\intcc{1/4,1}$ and assume $\load_{\frac{3}{4}n}(t)\leq\varnothing(t)$.
For the expected upper potential drop during round $t+1$ we have
\begin{equation}
\Ex{\Delta_+(t+1)|\bm{\load}(t)}\leq-\varepsilon\alpha\lambda\cdot\Phi_+(\bm{\load}(t))+2\alpha\lambda n
.
\end{equation}
\end{lemma}
The next lemma derives a bound that is used to bound the lower potential drop in reasonably balanced configurations.
\begin{lemma}\label{lem:bal_cfg:bound_lower_pot}
Consider a round $t$ and the constants $\varepsilon$ (from Claim~\ref{clm:bulk_leftright}) and $\alpha\leq\min(\ln(\sfrac{10}{9}),\varepsilon/8)$.
Let $\lambda\in\intcc{1/4,1}$ and assume $\load_{\frac{n}{4}}(t)\geq\varnothing(t)$.
For the expected lower potential drop during round $t$ we have
\begin{equation}
\Ex{\Delta_-(t+1)|\bm{\load}(t)}\leq-\varepsilon\alpha\lambda\cdot\Phi_-(\bm{\load}(t))+\frac{\alpha\lambda n}{2}
.
\end{equation}
\end{lemma}

The next lemma derives a bound that will be used to bound the potential drop in configurations with many balls far below the average to the right.
\begin{lemma}\label{lem:potdrop_unbalanced_upper}
Consider a round $t$ and constants $\alpha\leq1/46$ ($<\ln(10/9)$) and $\varepsilon\leq1/3$.
Let $\lambda\in\intcc{1/4,1}$ and assume $\load_{\frac{3}{4}n}(t)\geq\varnothing(t)$ and $\Ex{\Delta_+(t+1)|\bm{\load}(t)}\geq-\frac{\varepsilon\alpha\lambda}{4}\cdot\Phi_+(\bm{\load}(t))$.
Then we have either $\Phi_+(\bm{\load}(t))\leq\frac{\varepsilon}{4}\cdot\Phi_-(\bm{\load}(t))$ or $\Phi(\bm{\load}(t))=\varepsilon^{-8}\cdot\LDAUOmicron{n}$.
\end{lemma}

The next lemma derives a bound that will be used to bound the potential drop in configurations with many balls far above the average to the left.
\begin{lemma}\label{lem:potdrop_unbalanced_lower}
Consider a round $t$ and constants $\alpha\leq1/32$ ($<\ln(10/9)$) and $\varepsilon\leq1$.
Let $\lambda\in\intcc{1/4,1}$ and assume $\load_{\frac{n}{4}}(t)\leq\varnothing(t)$ and $\Ex{\Delta_-(t+1)|\bm{\load}(t)}\geq-\frac{\varepsilon\alpha\lambda}{4}\cdot\Phi_-(\bm{\load}(t))$.
Then we have either $\Phi_-(\bm{\load}(t))\leq\frac{\varepsilon}{4}\cdot\Phi_+(\bm{\load}(t))$ or $\Phi(\bm{\load}(t))=\varepsilon^{-8}\cdot\LDAUOmicron{n}$.
\end{lemma}
Putting all these lemmas together, we can derive the following bound on the potential change during a single round.
\begin{lemma}\label{lem:two_choice:onestep_phi_bound}
Consider an arbitrary round $t+1$ of the 2-choice process and the constants $\varepsilon$ (from Claim~\ref{clm:bulk_leftright}) and $\alpha\leq\min(\ln(10/9),\varepsilon/8)$.
For $\lambda\in\intcc{1/4,1}$ we have
\begin{equation}
\Ex{\Phi(\bm{\load}(t+1))|\bm{\load}(t)}\leq\left(1-\frac{\varepsilon\alpha\lambda}{4}\right)\cdot\Phi(\bm{\load}(t))+\varepsilon^{-8}\cdot\LDAUOmicron{n}
.
\end{equation}
\end{lemma}
We can use this result in a simple induction to prove Lemma~\ref{lem:two_choice:smoothness}.
\begin{proof}[Proof of Lemma~\ref{lem:two_choice:smoothness}]
Lemma~\ref{lem:two_choice:onestep_phi_bound} gives us a $\gamma<1$ and $c>0$ such that $\Ex{\Phi(\bm{\load}(t+1))|\bm{\load}(t)}\leq\gamma\cdot\Phi(\bm{\load}(t))+c$ holds for all rounds $t\geq0$.
Taking the expected value on both sides yields $\Ex{\Phi(\bm{\load}(t+1))}\leq\gamma\cdot\Ex{\Phi(\bm{\load}(t))}+c$.
Using induction and the linearity of the expected value, it is easy to check that $\Ex{\Phi(\bm{\load}(t))}\leq\frac{c}{1-\gamma}$ solves this recursion.
Using the values from Lemma~\ref{lem:two_choice:onestep_phi_bound} for $\gamma$ and $c$ (substituting $\varepsilon'$ for $\varepsilon$) we get $\Ex{\Phi(\bm{\load}(t))}\leq\frac{\smash{4\varepsilon'^{-8}}}{\varepsilon'\alpha\lambda}\cdot\LDAUOmicron{n}$.
The lemma's statement follows for the constant $\varepsilon=\LDAUOmicron{\varepsilon'^{-9}/(\alpha\lambda)}$.
\end{proof}

\subsection{Bounding the Maximum Load}\label{sec:two_choice:maxload}
The goal of this section is to prove Theorem~\ref{thm:two_choice:maxload}.
Remember the definitions of $\Phi(\bm{x})$ and $\Psi(\bm{x})$ from Equation~\eqref{eqn:two_choice:potentials}.
For any fixed round $t$, we will prove that (w.h.p.) $\Psi(\bm{\load}(t))=\LDAUOmicron{n\cdot\ln n}$, so that the average load is $\varnothing=\LDAUOmicron{\ln n}$.
Using a union bound and Lemma~\ref{lem:two_choice:smoothnessWHP}, we see that (w.h.p.) the the maximum load at the end of round $t$ is bounded by $\varnothing+\LDAUOmicron{\ln n}=\LDAUOmicron{\ln n}$.

It remains to prove a high probability bound on $\Psi(\bm{\load}(t))$ for arbitrary $t$.
To get an intuition for our analysis, consider the toy case $t=\poly(n)$ and assume that exactly $\lambda\cdot n\leq n$ balls are thrown each round.
Here, we can combine Observation~\ref{obs:two_choice:avg_distance} and Lemma~\ref{lem:two_choice:smoothness} to bound (w.h.p.) the load difference between any pair of bins and for all $t'<t$ by $\LDAUOmicron{\ln n}$ (via a union bound over $\poly(n)$ rounds).
Using the combinatorial observation that, while the load distance to the average is bounded by some $b\geq0$, the bound $\Psi\leq2b\cdot n$ is invariant under the 2-choice process (Lemma~\ref{lem:two_choice:total_load_for_small_phi}), we get for $b=\LDAUOmicron{\ln n}$ that $\Psi(\bm{\load}(t))\leq2b\cdot n=\LDAUOmicron{n\cdot \ln n}$, as required.
The case for $t=\LDAUomega{\poly(n)}$ is considerably more involved.
In particular, the fact that the number of balls in the system is only guaranteed to decrease when the total load is high \emph{and} the load distance to the average is low makes it challenging to design a suitable potential function that drops fast enough when it is high.
Thus, we deviate from this standard technique and elaborate on the idea of the toy case:
Instead of bounding (w.h.p.) the load difference between any pair of bins by $\LDAUOmicron{\ln n}$ for all $t'<t$ (which is not possible for $t\gg\poly(n)$), we prove (w.h.p.) an \emph{adaptive bound} of $\LDAUOmicron{\ln(t-t')\cdot f(\lambda)}$ for all $t'<t$, where $f$ is a suitable function (Lemma~\ref{lem:two_choice:adaptive_avg_bound}).
Then we consider the last round $t''<t$ with an empty bin.
Observation~\ref{obs:two_choice:avg_distance} yields a bound of $\Psi(\bm{\load}(t''))=2\cdot\LDAUOmicron{\ln(t-t'')\cdot f(\lambda)}\cdot n$ on the total load at time $t''$.
Using the same combinatorial observation as in the toy case, we get that (w.h.p.) $\Psi(\bm{\load}(t))\leq\Psi(\bm{\load}(t''))=2\cdot\LDAUOmicron{\ln(t-t'')\cdot f(\lambda)}\cdot n$.
The final step is to show that the load at time $t''$ (which is logarithmic in $t-t''$) decreases \emph{linearly} in $t-t''$, showing that the time interval $t-t''$ cannot be too large (or we would get a negative load at time $t$).
See Figure~\ref{fig:proofbypic}
for an illustration.

\begin{figure}
{\centering\includegraphics{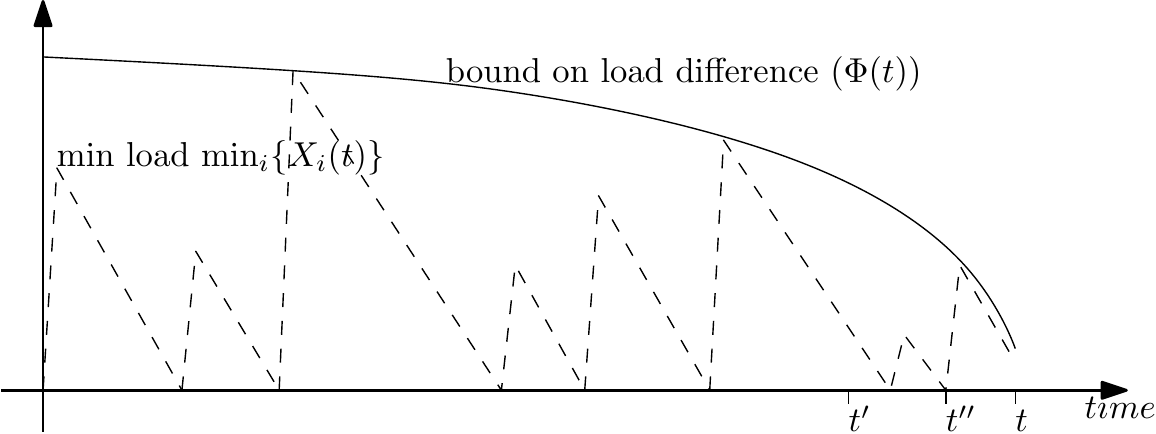}\\}
\caption{%
	To bound the system load at time $t$, consider the minimum load and our bound on the load difference over time.
	There was a last time $t''$ when there was an empty bin.
	The system load can only increase if there is an empty bin, and this increase is bounded by our bound on the load difference.
	Exploiting that the system load decreases linearly in time while every increase is bounded by our logarithmic bound on the load difference, we find a small interval $\intcc{t',t}$ containing $t''$.
}
\label{fig:proofbypic}
\end{figure}

\begin{lemma}\label{lem:two_choice:total_load_for_small_phi}
Let $b\geq0$ and consider a configuration $\bm{x}$ with $\Psi(\bm{x})\leq2b\cdot n$ and $\Phi(\bm{x})\leq e^{\alpha\cdot b}$.
Let $\bm{x'}$ denote the configuration after one step of the 2-choice process.
Then $\Psi(\bm{x'})\leq2b\cdot n$.
\end{lemma}
\begin{proof}
We distinguish two cases:
If there is no empty bin, then all $n$ bins delete one ball.
Since the maximum number of new balls is $n$, the number of balls cannot increase.
That is, we have $\Psi(\bm{x'})\leq\Psi(\bm{x})\leq2b\cdot n$.
Now consider the case that there is at least one empty bin.
Let $\eta\in\intoc{0,1}$ denote the fraction of empty bins (i.e., there are exactly $\eta\cdot n>0$ empty bins).
Since the minimal load is zero, Observation~\ref{obs:two_choice:avg_distance} implies $\max_ix_i\leq2b$.
Thus, the total number of balls in configuration $\bm{x}$ is at most $(1-\eta)n\cdot2b$.
Exactly $(1-\eta)n$ balls are deleted (one from each non-empty bin) and at most $n$ new balls enter the system.
We get $\Psi(\bm{x'})\leq(1-\eta)n\cdot2b-(1-\eta)n+n=(1-\eta)n\cdot(2b-1)+n\leq2b\cdot n$.
\end{proof}
\begin{lemma}\label{lem:two_choice:adaptive_avg_bound}
Let $\lambda\in\intco{1/4,1}$. Fix a round $t$.
For $i\in\N$ with $t-i\cdot\frac{8\log n}{1-\lambda}\geq0$ define $\mathcal{I}_i\coloneqq\intcc{t-i\cdot\frac{8\ln n}{1-\lambda},t}$.
Let $Y_i$ be the number of balls which spawn in $\mathcal{I}_i$.
\begin{enumerate}
\item Define the (good) smooth event $\mathcal{S}_t\coloneqq\bigcap_{t'<t}\bigl(\Phi(\bm{\load}(t'))\leq\abs{t-t'}^2\cdot n^2\bigr)$.
	Then $\Pr{\mathcal{S}_t}=1-\LDAUOmicron[small]{n^{-1}}$.
\item Define the (good) bounded balls event $\mathcal{B}_t\coloneqq\bigcap_i\bigl(Y_i\leq\frac{1+\lambda}{2}\cdot\abs{\mathcal{I}_i}\cdot n\bigr)$.
	Then $\Pr{\mathcal{B}_t}=1-\LDAUOmicron[small]{n^{-1}}$.
\end{enumerate}
\end{lemma}
\begin{proof}
Consider an arbitrary time $t'<t$.
By Lemma~\ref{lem:two_choice:smoothness} we have $\Ex{\Phi(t')}\leq n/\varepsilon$.
Using Markov's inequality, this implies $\pr{\Phi(t')\geq\abs{t-t'}^2\cdot n^2}\leq1/(\varepsilon\cdot\abs{t-t'}^2\cdot n)$.
Using the union bound over all $t'<t$ we calculate
\begin{equation}
\Pr{\bar{\mathcal{S}}_t}\leq\sum_{t'<t}\Pr{\Phi(t')\geq\abs{t-t'}^2\cdot n^2}\leq\frac{1}{\varepsilon n}\cdot\sum_{t'<t}\frac{1}{\abs{t-t'}^2}\leq\frac{\pi^2}{6\varepsilon\cdot n}=\LDAUOmicron{n^{-1}}
,
\end{equation}
where the last inequality applies the solution to the Basel problem.
This proves the first statement.

For the second statement, let $Z_i\coloneqq\abs{\mathcal{I}_i}\cdot n-Y_i$ be number of balls that did not spawn during $\mathcal{I}_i$.
Note that $Z_i$ is a sum of $\abs{\mathcal{I}_i}\cdot n$ independent indicator variables with $\Ex{Z_i}=(1-\lambda)\cdot\abs{\mathcal{I}_i}\cdot n=8i\cdot\ln n$.
Chernoff yields $\Pr{Z_i\leq(1-\lambda)\cdot\abs{\mathcal{I}_i}\cdot n/2}\leq e^{-8i\cdot\ln n/8}=n^{-i}. $
The desired statement follows by applying the identity $Z_i=\abs{\mathcal{I}_i}\cdot n-Y_i$ and taking the union bound.
\end{proof}
\begin{lemma}\label{lem:two_choice_totloadbound}
Fix a round $t$ and assume that both $\mathcal{S}_t$ and $\mathcal{B}_t$ hold.
Then $\Psi(\bm{\load}(t))\leq \frac{9n}{\alpha}\cdot\ln\bigl(\frac{n}{1-\lambda}\bigr)$.
\end{lemma}
\begin{proof}
Let $t'<t$ be the last time when there was an empty bin and set $\Delta\coloneqq t-t'$.
Note that $t'$ is well defined, as we have $X_i(0)=0$ for all $i\in\intcc{n}$.
Since $\mathcal{S}_t$ holds, we have $\Phi(\bm{\load}(t'))\leq\Delta^2\cdot n^2=\exp\left(\ln(\Delta^2\cdot n^2)\right)$.
By choice of $t'$ we have $\min_i\load_i(t')=0$.
Together with Observation~\ref{obs:two_choice:avg_distance} we get $\max_i\load_i(t'))\leq2\ln\bigl(\Delta^2\cdot n^2\bigr)/\alpha$.
Summing up over all bins (and pulling out the square), this implies $\Psi(\bm{\load}(t'))\leq4n\cdot\ln\bigl(\Delta\cdot n\bigr)/\alpha$.
Applying Lemma~\ref{lem:two_choice:total_load_for_small_phi} yields $\Psi(\bm{\load}(t'+1))\leq4n\cdot\ln\bigl(\Delta\cdot n\bigr)/\alpha$.
By choice of $t'$, there is no empty bin in $\bm{\load}(t'')$ for all $t''\in\set{t'+1,t'+2,\dots,t-1}$.
Thus, during each of these rounds exactly $n$ balls are deleted.
To bound the number of deleted balls, let $i$ be maximal with $\mathcal{I}_i\subseteq\intcc{t',t}$  (as defined in Lemma~\ref{lem:two_choice:adaptive_avg_bound}).
Since $\mathcal{B}_t$ holds and using the maximality of $i$, the number of balls $Y$ that spawn during $\intcc{t',t}$ is at most $(1+\lambda)\abs{\mathcal{I}_i}\cdot n/2+\frac{8\ln n}{1-\lambda}\cdot n\leq(1+\lambda)\Delta\cdot n/2+\frac{8\ln n}{1-\lambda}\cdot n$.
We calculate
\begin{equation}\label{eqn:two_choice_totloadbound:loadbound}
\begin{aligned}
\Psi(\bm{\load}(t))
&\leq \Psi(\bm{\load}(t'+1))-\Delta\cdot n+Y
 \leq \frac{4n}{\alpha}\cdot\ln(\Delta\cdot n)-\frac{1-\lambda}{2}\Delta\cdot n+\frac{8\ln n}{1-\lambda}\cdot n\\
&=    \frac{1-\lambda}{2}\cdot n\cdot\left(\frac{8}{\alpha(1-\lambda)}\cdot\ln(\Delta\cdot n)-\Delta+\frac{16 \ln n}{(1-\lambda)^2}\right)\\
&\leq \frac{1-\lambda}{2}\cdot\Delta\cdot n\cdot\left(\frac{24}{\alpha(1-\lambda)^2}\cdot\frac{\ln(\Delta\cdot n)}{\Delta}-1\right).
\end{aligned}
\end{equation}
With $f=f(\lambda)\coloneqq24/\bigl(\alpha(1-\lambda)^2\bigr)$ the last factor becomes $f\cdot\ln(\Delta\cdot n)/\Delta-1$.
It is negative if and only if $\Delta>f\cdot\ln(\Delta\cdot n)$.
This inequality holds for any $\Delta>-f\cdot W_{-1}(-\smash{\frac{1}{f\cdot n}})$, where $W_{-1}$ denotes the lower branch of the Lambert W function\footnote{%
	Note that $-\frac{1}{f\cdot n}\geq-1/e$, so that $W_{-1}(-\frac{1}{f\cdot n})$ is well defined.}.
This implies that $\Delta\leq-f\cdot W_{-1}(-\sfrac{1}{fn})$, since otherwise we would have $\Psi(\bm{\load}(t))<0$, which is clearly a contradiction.
Using the Taylor approximation $W_{-1}(x)=\ln(-x)-\ln\bigl(\ln(-1/x)\bigr)-\LDAUomicron{1}$ as $x\to-0$, we get
\begin{equation}
\Delta\leq-f\cdot W_{-1}\left(-\frac{1}{f\cdot n}\right)\leq f\cdot\ln(f\cdot n)+f\cdot\ln\bigl(\ln(f\cdot n)\bigr)+f\leq2 f \cdot\ln(f\cdot n)
.
\end{equation}
Finally, we use this bound on $\Delta$ to get
\begin{equation*}
\begin{aligned}
\Psi(\bm{\load}(t))
&\leq\Psi(\bm{\load}(t'+1)\leq\frac{4n}{\alpha}\cdot\ln(\Delta\cdot n)\leq\frac{4n}{\alpha}\cdot\ln\bigl(2fn\cdot\ln(f n)\bigr)\\
&\leq \frac{4n}{\alpha}\cdot\ln\left(\frac{48n}{\alpha(1-\lambda)^2}\cdot\ln\left(\frac{24n}{\alpha(1-\lambda)^2}\right)\right)\leq\frac{9n}{\alpha}\cdot\ln\left(\frac{n}{1-\lambda}\right)
.
\end{aligned}
\qedhere
\end{equation*}
\end{proof}
Now, by combining Lemma~\ref{lem:two_choice_totloadbound} with the fact that the events $\mathcal{S}_t$ and $\mathcal{B}_t$ hold with high probability (Lemma~\ref{lem:two_choice:adaptive_avg_bound}), we immediately get that (w.h.p.) $\Psi(\bm{\load}(t))=\LDAUOmicron{n\cdot\ln n}$.
As described at the beginning of this section, combining this with Lemma~\ref{lem:two_choice:smoothnessWHP} proves Theorem~\ref{thm:two_choice:maxload}.

\subsection{Stability}\label{sec:two_choice:stability}
This section proves Theorem~\ref{thm:two_choice:stability}.
The first auxiliary lemma states that for sufficiently high value of $\Gamma$, this potential decreases.\footnote{%
	It might look tempting to use $\Gamma$ together with Hajek's theorem to derive a bound on the maximum load of system.
	However, this would require (exponentially) sharper bounds on $\Phi$.
}
\begin{lemma}[Negative Bias $\Gamma$]\label{GammaDrop}
Let $\lambda\in\intco{1/4,1}$.
If $\Gamma(\bm{\load}(t))\geq 2\tfrac{n^4}{(1-\lambda)^2\lambda}$, then $$\Ex{\Gamma(\bm{\load}(t+1))-\Gamma(\bm{\load}(t))|\bm{\load}(t)}\leq -1.$$
\end{lemma}

\begin{proof}
Assume $\bm{\load}(t)=x$ is fixed.
By definition of $\Gamma(\cdot )$,
we have $\Phi(x)\geq \Gamma(x)/2$ or $\Psi(x)\geq \Gamma(x)/2$. We now show that in both cases $\Ex{\Gamma(\bm{\load}(t+1))-\Gamma(x)|\bm{\load}(t)=x}\leq -1$.
\begin{enumerate}
\item If $\Phi(x)\geq \Gamma(x)/2$, then we have, by Lemma~\ref{lem:two_choice:onestep_phi_bound}, a potential drop of $$\Ex{\Phi(\bm{\load}(t+1))-\Phi(x)|\bm{\load}(t)=x}\leq-(\varepsilon\alpha\lambda/4)\cdot\Phi(x)+n\log n\leq -(\varepsilon\alpha\lambda/8)\cdot\Gamma(x)+n\log n.$$
Note that, by definition of $\Psi$,  $\Psi(\bm{\load}(t+1))-\Psi(x)\leq n.$
Together with  $\Gamma(x) \geq \tfrac{8(n\log n+n^2/(1-\lambda) +1)}{e\alpha \lambda }$, $$\Ex{\Gamma(\bm{\load}(t+1))-\Gamma(x)|\bm{\load}(t)=x}\leq-(\varepsilon\alpha\lambda/8)\cdot\Gamma(x)+n\log n +(n/(1-\lambda))\cdot n \leq -1.$$

\item Otherwise, i.e., if $\Phi(x)< \Gamma(x)/2$,  we have that
\begin{itemize}
\item[(i)] the load difference is, by Observation~\ref{obs:two_choice:avg_distance}, bounded by $2\ln(\Gamma(x)/2)/\alpha$, and
\item[(ii)] $\Psi(x)\geq \Gamma(x)/2$ must hold. This implies that $\varnothing \geq \frac{1}{n}\left( \tfrac{\Gamma(x)/2}{\frac{n}{1-\lambda}}\right)=\tfrac{(1-\lambda)\cdot\Gamma(x)}{2n^2 }$.
\end{itemize}
From $(i)$ and $(ii)$ we have that the minimum load is at least $\tfrac{(1-\lambda)\cdot\Gamma(x)}{2n^2 } - \ln(\Gamma(x)/2)/\alpha$.
From Lemma~\ref{xminuslnx} and
$\Gamma(x)\geq 2\tfrac{n^4}{(1-\lambda)^2\lambda}$, it follows that every bin has load at least load $1$.
Thus each bin will delete one ball and the number of  balls arriving is $\lambda n$ in expectation.
Hence, $\Ex{\Psi(\bm{\load}(t+1))-\Psi(x)|\bm{\load}(t)=x}= -\frac{n}{1-\lambda}(1-\lambda)n$.
Now,

\begin{align}
\begin{split}
\Ex{\Gamma(\bm{\load}(t+1))-\Gamma(x)|\bm{\load}(t)=x} &=  \Ex{\Phi(\bm{\load}(t+1))-\Phi(x)|\bm{\load}(t)=x} -\frac{n}{1-\lambda}(1-\lambda)n\\
&\leq
n\log n -\frac{n}{1-\lambda}(1-\lambda)n \leq -1.
\end{split}
\end{align}

\end{enumerate}
Thus, $\Ex{\Gamma(\bm{\load}(t+1))-\Gamma(x)|\bm{\load}(t)=x}\leq -1$, which yields the claim.
\end{proof}

\begin{lemma}\label{xminuslnx}
For all $x \geq 2\tfrac{n^4}{(1-\lambda)^2\lambda}$ it holds that 
$\tfrac{(1-\lambda)\cdot x}{2n^2 } - 2\ln(x/2)/\alpha\geq 1$.
\end{lemma}
\begin{proof}
Define $f(x)=\tfrac{(1-\lambda)\cdot x}{2n^2 } - 2\ln(x/2)/\alpha$. 
We have
$
f\left(2\tfrac{n^4}{(1-\lambda)^2\lambda}\right)
\geq \tfrac{n^2 }{(1-\lambda) \lambda } -\tfrac{2}{\alpha}\ln\left( \frac{n^4}{(1-\lambda)^2 \lambda }  \right)
\geq 1,
$
where the last inequality holds for large enough of $n$ since $\alpha$ is a constant.
Moreover, for all $x \geq 2\tfrac{n^4}{(1-\lambda)^2\lambda}$ we have  $f'(x)=\tfrac{1-\lambda}{n^2 } - \tfrac{2}{\alpha x} \geq 0$.
Thus, the claim follows.
\end{proof}

We are ready to prove Theorem~\ref{thm:two_choice:stability}.
\begin{proof}[Proof of Theorem~\ref{thm:two_choice:stability}]
The proof proceeds by applying Theorem~\ref{thm:foster}. 
We now define the parameters of Theorem~\ref{thm:foster}.
Let $\zeta(t)=\bm{\load}(t)$ and hence $\Omega$ is the state space of $\load$. 
First we observe that $\Omega$ is countable since there are a constant number of bins ($n$ is consider a constant in this matter) each having a load which is a natural number.
We define $\phi(\bm{\load}(t))$ to be  $\Gamma(\bm{\load}(t))$.
We define  $C=\{ x : \Gamma(x) \leq 2\tfrac{n^4}{(1-\lambda)^2\lambda}\}$. 
Define $\beta(x)=1$ and $\eta = 1$.
We now show that the preconditions (a) and (b) of Theorem~\ref{thm:foster} are fulfilled.
\begin{itemize}
\item Let $x \not\in C$. By definition of $C$ and $\phi(\bm{\load}(t))$, and from Lemma~\ref{GammaDrop} we have 
\begin{align}\label{final12}
E[\phi(\load({t+1})) - \phi (x) | \bm{\load}(t) = x] \leq E[\Gamma(\bm{\load}(t+1)) - \Gamma(x) | \bm{\load}(t) = x] \leq -1.
\end{align}
\item Let $x \in C$. Recall that $\Gamma(\bm{\load}(t))=\Phi(\bm{\load}(t))+\Psi(\bm{\load}(t))$. By Lemma~\ref{lem:potdrop_unbalanced_lower} and the fact the the number of balls arriving in one round is bounded by $n$, we derive,
\begin{align}\label{final22}
E[\phi(\load({t+1})) | \bm{\load}(t) = x] &=  E[\Phi(\bm{\load}(t+1)) | \bm{\load}(t) = x]+E[\Psi(\bm{\load}(t+1)) | \bm{\load}(t) = x]\notag\\
&\leq \left(\left(1-\frac{\varepsilon\alpha\lambda}{4}\right)2\tfrac{n^4}{(1-\lambda)^2\lambda}\right) + \frac{n}{1-\lambda}n  < \infty.
\end{align}

\end{itemize}
The claim follows by applying Theorem~\ref{thm:foster} with Equations~\eqref{final12} and~\eqref{final22}. 
\end{proof}


\bibliographystyle{abbrvnat}
\bibliography{dblp}

\begin{thebibliography}{24}
\providecommand{\natexlab}[1]{#1}
\providecommand{\url}[1]{\texttt{#1}}
\expandafter\ifx\csname urlstyle\endcsname\relax
  \providecommand{\doi}[1]{doi: #1}\else
  \providecommand{\doi}{doi: \begingroup \urlstyle{rm}\Url}\fi

\bibitem[Adler et~al.(1998{\natexlab{a}})Adler, Berenbrink, and
  Schr\"{o}der]{ABS98}
M.~Adler, P.~Berenbrink, and K.~Schr\"{o}der.
\newblock Analyzing an infinite parallel job allocation process.
\newblock In \emph{Proceedings of the 6th Annual European Symposium on
  Algorithms}, ESA '98, pages 417--428, London, UK, UK, 1998{\natexlab{a}}.
  Springer-Verlag.
\newblock ISBN 3-540-64848-8.
\newblock URL \url{http://dl.acm.org/citation.cfm?id=647908.740138}.

\bibitem[Adler et~al.(1998{\natexlab{b}})Adler, Chakrabarti, Mitzenmacher, and
  Rasmussen]{ACMR98}
M.~Adler, S.~Chakrabarti, M.~Mitzenmacher, and L.~Rasmussen.
\newblock Parallel randomized load balancing.
\newblock \emph{Random Structures \& Algorithms}, 13\penalty0 (2):\penalty0
  159--188, 1998{\natexlab{b}}.
\newblock ISSN 1098-2418.
\newblock \doi{10.1002/(SICI)1098-2418(199809)13:2<159::AID-RSA3>3.0.CO;2-Q}.
\newblock URL
  \url{http://dx.doi.org/10.1002/(SICI)1098-2418(199809)13:2<159::AID-RSA3>3.0.CO;2-Q}.

\bibitem[Alfa(2003)]{Alfa03}
A.~S. Alfa.
\newblock Algorithmic analysis of the bmap/d/k system in discrete time.
\newblock \emph{Adv. in Appl. Probab.}, 35\penalty0 (4):\penalty0 1131--1152,
  12 2003.
\newblock \doi{10.1239/aap/1067436338}.
\newblock URL \url{http://dx.doi.org/10.1239/aap/1067436338}.

\bibitem[Azar et~al.(1999)Azar, Broder, Karlin, and Upfal]{ABKU99}
Y.~Azar, A.~Z. Broder, A.~R. Karlin, and E.~Upfal.
\newblock Balanced allocations.
\newblock \emph{SIAM Journal on Computing}, 29\penalty0 (1):\penalty0 180--200,
  1999.
\newblock \doi{10.1137/S0097539795288490}.

\bibitem[Becchetti et~al.(2015)Becchetti, Clementi, Natale, Pasquale, and
  Posta]{BCNPP15}
L.~Becchetti, A.~E.~F. Clementi, E.~Natale, F.~Pasquale, and G.~Posta.
\newblock Self-stabilizing repeated balls-into-bins.
\newblock In G.~E. Blelloch and K.~Agrawal, editors, \emph{Proceedings of the
  27th {ACM} on Symposium on Parallelism in Algorithms and Architectures,
  {SPAA} 2015, Portland, OR, USA, June 13-15, 2015}, pages 332--339. {ACM},
  2015.
\newblock ISBN 978-1-4503-3588-1.
\newblock \doi{10.1145/2755573.2755584}.
\newblock URL \url{http://doi.acm.org/10.1145/2755573.2755584}.

\bibitem[Berenbrink et~al.(2000)Berenbrink, Czumaj, Friedetzky, and
  Vvedenskaya]{BCFV00}
P.~Berenbrink, A.~Czumaj, T.~Friedetzky, and N.~D. Vvedenskaya.
\newblock Infinite parallel job allocation (extended abstract).
\newblock In \emph{Proceedings of the Twelfth Annual ACM Symposium on Parallel
  Algorithms and Architectures}, SPAA '00, pages 99--108, New York, NY, USA,
  2000. ACM.
\newblock ISBN 1-58113-185-2.
\newblock \doi{10.1145/341800.341813}.
\newblock URL \url{http://doi.acm.org/10.1145/341800.341813}.

\bibitem[Berenbrink et~al.(2006)Berenbrink, Czumaj, Steger, and
  V{\"o}cking]{BCSV06}
P.~Berenbrink, A.~Czumaj, A.~Steger, and B.~V{\"o}cking.
\newblock Balanced allocations: The heavily loaded case.
\newblock \emph{SIAM Journal on Computing}, 35\penalty0 (6):\penalty0
  1350--1385, 2006.
\newblock \doi{10.1137/S009753970444435X}.

\bibitem[Berenbrink et~al.(2012)Berenbrink, Czumaj, Englert, Friedetzky, and
  Nagel]{BCEFN12}
P.~Berenbrink, A.~Czumaj, M.~Englert, T.~Friedetzky, and L.~Nagel.
\newblock Multiple-choice balanced allocation in (almost) parallel.
\newblock In A.~Gupta, K.~Jansen, J.~Rolim, and R.~Servedio, editors,
  \emph{Approximation, Randomization, and Combinatorial Optimization.
  Algorithms and Techniques}, volume 7408 of \emph{Lecture Notes in Computer
  Science}, pages 411--422. Springer Berlin Heidelberg, 2012.
\newblock ISBN 978-3-642-32511-3.
\newblock \doi{10.1007/978-3-642-32512-0_35}.
\newblock URL \url{http://dx.doi.org/10.1007/978-3-642-32512-0_35}.

\bibitem[Berenbrink et~al.(2013)Berenbrink, Khodamoradi, Sauerwald, and
  Stauffer]{BKSS13}
P.~Berenbrink, K.~Khodamoradi, T.~Sauerwald, and A.~Stauffer.
\newblock Balls-into-bins with nearly optimal load distribution.
\newblock In \emph{Proceedings of the Twenty-fifth Annual ACM Symposium on
  Parallelism in Algorithms and Architectures}, SPAA '13, pages 326--335, New
  York, NY, USA, 2013. ACM.
\newblock ISBN 978-1-4503-1572-2.
\newblock \doi{10.1145/2486159.2486191}.
\newblock URL \url{http://doi.acm.org/10.1145/2486159.2486191}.

\bibitem[Czumaj(1998)]{Czumaj98}
A.~Czumaj.
\newblock Recovery time of dynamic allocation processes.
\newblock In \emph{Proceedings of the Tenth Annual ACM Symposium on Parallel
  Algorithms and Architectures}, SPAA '98, pages 202--211, New York, NY, USA,
  1998. ACM.
\newblock ISBN 0-89791-989-0.
\newblock \doi{10.1145/277651.277686}.
\newblock URL \url{http://doi.acm.org/10.1145/277651.277686}.

\bibitem[Czumaj and Stemann(1997)]{CS97}
A.~Czumaj and V.~Stemann.
\newblock Randomized allocation processes.
\newblock In \emph{Foundations of Computer Science, 1997. Proceedings., 38th
  Annual Symposium on}, pages 194--203, Oct 1997.
\newblock \doi{10.1109/SFCS.1997.646108}.

\bibitem[Fayolle et~al.(1995)Fayolle, Malyshev, and Menshikov]{FMM95}
G.~Fayolle, V.~Malyshev, and M.~Menshikov.
\newblock \emph{Topics in the Constructive Theory of Countable Markov Chains}.
\newblock Cambridge University Press, 1995.
\newblock ISBN 9780521461979.
\newblock URL \url{https://books.google.ca/books?id=lTJltFEnnHcC}.

\bibitem[Gonnet(1981)]{Gonnet81}
G.~H. Gonnet.
\newblock Expected length of the longest probe sequence in hash code searching.
\newblock \emph{J. ACM}, 28\penalty0 (2):\penalty0 289--304, Apr. 1981.
\newblock ISSN 0004-5411.
\newblock \doi{10.1145/322248.322254}.
\newblock URL \url{http://doi.acm.org/10.1145/322248.322254}.

\bibitem[Hajek()]{H82}
B.~Hajek.
\newblock Hitting-time and occupation-time bounds implied by drift analysis
  with applications.
\newblock \emph{Advances in Applied Probability}, 14\penalty0 (3):\penalty0
  502--525.

\bibitem[Kamal(1996)]{Kamal96}
A.~Kamal.
\newblock Efficient solution of multiple server queues with application to the
  modeling of atm concentrators.
\newblock In \emph{INFOCOM '96. Fifteenth Annual Joint Conference of the IEEE
  Computer Societies. Networking the Next Generation. Proceedings IEEE},
  volume~1, pages 248--254 vol.1, Mar 1996.
\newblock \doi{10.1109/INFCOM.1996.497900}.

\bibitem[Kim et~al.(2012)Kim, Chaudhry, Yoon, and Kim]{Kim12}
N.~K. Kim, M.~L. Chaudhry, B.~K. Yoon, and K.~Kim.
\newblock A complete and simple solution to a discrete-time finite-capacity
  bmap/d/c queue.
\newblock 2012.

\bibitem[Levin and Perres(2008)]{Levin:2008}
D.~A. Levin and Y.~Perres.
\newblock \emph{Markov Chains and Mixing Times}.
\newblock American Mathematical Society, December 2008.
\newblock ISBN 978-0-8218-4739-8.

\bibitem[Mitzenmacher(2001)]{Mitzenmacher01}
M.~Mitzenmacher.
\newblock The power of two choices in randomized load balancing.
\newblock \emph{{IEEE} Trans. Parallel Distrib. Syst.}, 12\penalty0
  (10):\penalty0 1094--1104, 2001.
\newblock \doi{10.1109/71.963420}.
\newblock URL \url{http://doi.ieeecomputersociety.org/10.1109/71.963420}.

\bibitem[Peres et~al.(2010)Peres, Talwar, and
  Wieder]{Peres:2010:CPW:1873601.1873732}
Y.~Peres, K.~Talwar, and U.~Wieder.
\newblock The ($1+\beta$)-choice process and weighted balls-into-bins.
\newblock In \emph{Proceedings of the 21st Annual ACM-SIAM Symposium on
  Discrete Algorithms (SODA)}, SODA '10, pages 1613--1619, Philadelphia, PA,
  USA, 2010. Society for Industrial and Applied Mathematics.
\newblock ISBN 978-0-898716-98-6.

\bibitem[Raab and Steger(1998)]{RS98}
M.~Raab and A.~Steger.
\newblock "balls into bins" - {A} simple and tight analysis.
\newblock In M.~Luby, J.~D.~P. Rolim, and M.~J. Serna, editors,
  \emph{Randomization and Approximation Techniques in Computer Science, Second
  International Workshop, RANDOM'98, Barcelona, Spain, October 8-10, 1998,
  Proceedings}, volume 1518 of \emph{Lecture Notes in Computer Science}, pages
  159--170. Springer, 1998.
\newblock ISBN 3-540-65142-X.
\newblock \doi{10.1007/3-540-49543-6_13}.
\newblock URL \url{http://dx.doi.org/10.1007/3-540-49543-6_13}.

\bibitem[Sohraby and Zhang(1992)]{Sohraby92}
K.~Sohraby and J.~Zhang.
\newblock Spectral decomposition approach for transient analysis of
  multi-server discrete-time queues.
\newblock In \emph{INFOCOM'92. Eleventh Annual Joint Conference of the IEEE
  Computer and Communications Societies, IEEE}, pages 395--404. IEEE, 1992.

\bibitem[Stemann(1996)]{Stemann96}
V.~Stemann.
\newblock Parallel balanced allocations.
\newblock In \emph{Proceedings of the Eighth Annual ACM Symposium on Parallel
  Algorithms and Architectures}, SPAA '96, pages 261--269, New York, NY, USA,
  1996. ACM.
\newblock ISBN 0-89791-809-6.
\newblock \doi{10.1145/237502.237565}.
\newblock URL \url{http://doi.acm.org/10.1145/237502.237565}.

\bibitem[Talwar and Wieder(2013)]{TalwarW13}
K.~Talwar and U.~Wieder.
\newblock Balanced allocations: {A} simple proof for the heavily loaded case.
\newblock \emph{CoRR}, abs/1310.5367, 2013.
\newblock URL \url{http://arxiv.org/abs/1310.5367}.

\bibitem[Talwar and Wieder(2014)]{TW14}
K.~Talwar and U.~Wieder.
\newblock Balanced allocations: A simple proof for the heavily loaded case.
\newblock In J.~Esparza, P.~Fraigniaud, T.~Husfeldt, and E.~Koutsoupias,
  editors, \emph{Automata, Languages, and Programming}, volume 8572 of
  \emph{Lecture Notes in Computer Science}, pages 979--990. Springer Berlin
  Heidelberg, 2014.
\newblock ISBN 978-3-662-43947-0.
\newblock \doi{10.1007/978-3-662-43948-7_81}.
\newblock URL \url{http://dx.doi.org/10.1007/978-3-662-43948-7_81}.

\end{thebibliography}

\clearpage

\appendix


\section{Auxiliary Results}\label{app:auxiliary}
\begin{theorem}[{\textcite[Theorem~2.2.4]{FMM95}}]\label{thm:foster} 
A time-homogeneous irreducible aperiodic Markov chain $\zeta$ with a countable state space $\Omega$ is positive recurrent if and only if there exists a positive function $\phi(x), x\in \Omega$, a number $\eta > 0$, a positive integer-valued function $\beta(x), x \in \Omega$, and a finite set $C \subseteq \Omega$ such that the following inequalities hold:
\begin{enumerate}
\item\label{thm:foster:a} $E[\phi(\zeta({t+\beta(x)})) - \phi (x) | \zeta(t) = x] \leq -\eta \beta(x),$ $x\not\in C$
\item\label{thm:foster:b} $E[\phi(\zeta({t+\beta(x)}))  | \zeta(t) = x] < \infty,$ $x\in C$
\end{enumerate}
\end{theorem}

\begin{theorem}[Simplified version of {\textcite[Theorem~2.3]{H82}}]\label{thm:Hajek}
Let $(Y(t))_{t \geq 0}$ be a sequence of random variables on a probability space $(\Omega, \mathcal{F}, P)$ with respect to the filtration $(\mathcal{F}(t))_{t\geq 0}$.  
Assume the following two conditions hold:
\begin{enumerate}
\item[(i)] (Majorization) There exists a random variable $Z$ and a constant $\lambda' > 0$,
such that $\ex{e^{\lambda' Z}}\leq D$ for some finite $D$, and $(|Y({t+1}) - Y(t)| \big\vert
\mathcal{F}(t))   \prec Z$ for all $t \geq 0$; and
\item[(ii)] (Negative Bias) There exist $a,\varepsilon_0 > 0$, such for all $t$ we have
$$\ex{Y({t+1})-Y(t)|\mathcal{F}(t), Y(t) > a } \leq -\varepsilon_0. $$
\end{enumerate}
Let  $\eta = \min \{ \lambda', \varepsilon_0\cdot \lambda'^2/(2D), 1/(2\varepsilon_0) \}$.
Then,  for all $b $ and $t$ we have 
$$ \Pr{Y(t) \geq b | \mathcal{F}(0)} \leq  e^{\eta (Y(0) -  b)}+ \frac{2 D}{\varepsilon_0 \cdot \eta}\cdot e^{ \eta  (a -  b)}. $$
\end{theorem}
\begin{proof}
The statement of the theorem provided in \cite{H82} requires besides  $(i)$ and $(ii)$ to
choose constants $\eta$, and $\rho$ such that 
$0< \rho \leq  \lambda'$, $\eta < \varepsilon_0/c$
 and $\rho=1-\varepsilon_0\cdot \eta +c \eta^2$
where $c=\frac{\Ex{e^{\lambda' Z}} - (1+ \lambda' \Ex{Z})}{\lambda'^2}=\sum_{k=2}^\infty \frac{\lambda'^{k-2}}{k!}\Ex{Z^k}.$
With these requirements it then holds that for all $b$ and $t$ 

\begin{align}\label{Bildschirm}
\begin{split}
\Pr{Y(t) \geq b | \mathcal{F}(0)} \leq \rho^t e^{\eta (Y(0) -  b)}+ \frac{1-\rho^t}{1-\rho}\cdot D\cdot  e^{\eta(a -  b)}.
\end{split}
\end{align}

In the following we bound \eqref{Bildschirm} by setting
 $\eta = \min \{ \lambda', \varepsilon_0\cdot \lambda'^2/(2D), 1/(2\varepsilon_0) \}$.
The following upper and lower bound on $\rho$ follow.
\begin{itemize}
\item $\rho=1-\varepsilon_0 \cdot \eta+c \eta^2 \leq 1-\varepsilon_0 \cdot \eta+ \varepsilon_0 \cdot \eta\cdot c \cdot \lambda'^2/(2D) \leq   1-\varepsilon_0 \cdot \eta+ \varepsilon_0 \cdot \eta/2 =1-\varepsilon_0 \cdot \eta/2
,$
where we used  $c\leq D/\lambda'^2$.

\item $\rho=1-\varepsilon_0 \cdot \eta+c \eta^2  \geq 1-\varepsilon_0 /(2\varepsilon_0)  \geq 0$. 
\end{itemize}


We derive, from \eqref{Bildschirm} using that for any $t\geq 0$ we have $0\leq \rho^t\leq 1$

\begin{align}
\begin{split}
      \Pr{Y(t) \geq b | \mathcal{F}(0)}
&\leq \rho^t e^{\eta (Y(0) -  b)}+ \frac{1-\rho^t}{1-\rho}\cdot D\cdot  e^{\eta(a -  b)}
 \leq e^{\eta (Y(0) -  b)}+ \frac{1}{1-\rho}\cdot D\cdot  e^{\eta(a -  b)}\\
&\leq e^{\eta (Y(0) -  b)}+ \frac{2 D}{\varepsilon_0 \cdot \eta}\cdot e^{ \eta  (a -  b)},
\end{split}
\end{align}
since $\frac{1}{(1-\rho)}\leq
\frac{2 }{\varepsilon_0 \cdot \eta}$.  This yields the claim.
\end{proof}

\begin{theorem}[{\textcite[Theorem~1]{RS98}}]\label{kuchen}
Let $M$ be the random variable that counts the maximum number
of balls in any bin, if we throw m balls independently and uniformly at random into $n$ bins. Then
$\Pr{M > k_\alpha}=o(1)$ if $\alpha > 1$ and $\Pr{M>k_\alpha}=1-o(1)$ if $0< \alpha < 1$, where

$$k_\alpha=
\begin{cases}
\tfrac{\log n}{\log \tfrac{n \log n}{m}}\left(1+ \alpha \tfrac{\log\log \tfrac{n\log n}{m}}{\log \tfrac{n \log n}{m}} \right) & if\ \tfrac{n}{\polylog(n)} \leq m \ll n \log n\\
(d_c -1 +\alpha) \log n& if\ m=c\cdot n\log n\text{ for some constant $c$}\\
\tfrac{m}{n} + \alpha\sqrt{2\tfrac{m}{n}\log n} & if\ n \log n \ll m \leq n \polylog(n)\\
\tfrac{m}{n} + \sqrt{2\tfrac{m}{n}\log n\left(1-\tfrac{1}{\alpha}\tfrac{\log \log n}{2\log n} \right)} & if\ m \gg n (\log n)^3,
\end{cases}
 $$
where $d_c$ is largest solution of
$ 1 + x (\log c - \log x + 1) -c = 0$.
We have $d_1=e$ and $d_{1.00001}= 2.7183$. 
\end{theorem}

\section{Auxiliary Tools for the 2-Choice Process}\label{app:aux:two_choice}
\begin{claim}\label{clm:two_choice:delta_bounds}
Consider a bin $i$ and the values $\hat{\delta}_i$ and $\check{\delta}_i$ as defined before Lemma~\ref{lem:pot_change}.
If $\alpha\leq\ln(10/9)$, then $\max(\abs{\hat{\delta}_i},\abs{\check{\delta}_i})\leq\frac{5}{4}\lambda$.
\end{claim}
\begin{proof}
Remember that $\hat{\delta}_i\coloneqq\lambda n\cdot(\sfrac{1}{n}\cdot\check{1}-p_i\cdot\sfrac{\hat{\alpha}}{\alpha})$ and $\check{\delta}_i\coloneqq\lambda n\cdot(\sfrac{1}{n}\cdot\hat{1}-p_i\cdot\sfrac{\check{\alpha}}{\alpha})$, where $\check{1}=1-\alpha/n<1<1+\alpha/n=\hat{1}$ (see proof of Lemma~\ref{lem:pot_change}).
Note that if $\alpha\leq\ln(10/9)$, we have $\hat{1}<5/4$ and $\check{1}>8/9$.
The claims hold trivially for $i=1$, since then $p_i=(2i-1)/n^2=1/n^2$ and both $\abs{\sfrac{1}{n}\cdot\check{1}-p_i\cdot\sfrac{\hat{\alpha}}{\alpha}}\leq1/n$ and $\abs{\sfrac{1}{n}\cdot\hat{1}-p_i\cdot\sfrac{\check{\alpha}}{\alpha}}\leq\hat{1}/n$.
For the other extreme, $i=n$, we have $p_n\leq2/n$.
The first statement follows from this and the definition of $\hat{\alpha}=e^{\alpha}-1$ since $\frac{2}{n}\cdot\frac{\hat{\alpha}}{\alpha}-\frac{1}{n}\cdot\check{1}\leq\frac{2}{n}\frac{10/9-1}{\ln(10/9)}-\frac{1}{n}\cdot\check{1}<\frac{5}{4n}$.
Similarly, the second statement follows together with $\frac{2}{n}\frac{\check{\alpha}}{\alpha}-\frac{1}{n}\cdot\hat{1}<\frac{1}{n}$ (which holds for any $\alpha>0$).
\end{proof}

\begin{claim}\label{clm:bulk_leftright}
There is a constant $\varepsilon>0$ such that
\begin{equation}\label{heavy34}
\sum_{i\leq\frac{3}{4}n}p_i\cdot\Phi_{i,+}\leq(1-2\varepsilon)\cdot\frac{\Phi_+}{n}
.
\end{equation}
and
\begin{equation}\label{heavy14}
\sum_{i\in\intcc{n}}p_i\cdot\Phi_{i,-}\geq(1+2\varepsilon)\cdot\frac{\Phi_--\sum_{i\leq\frac{n}{4}}\Phi_{i,-}}{n}
.
\end{equation}
\end{claim}
\begin{proof}
The claim follows from comments in~\cite{TW14}.
For Equation~\ref{heavy34} recall that $\sum_{i<3n/4}\Phi_{i,+} \leq \Phi_+$ (by definition).
Since $\Phi_{i,+}$ for $i = 1, \dots, n$ is non-increasing where $i$ is the $i$-th loaded bin, the above equation is maximized where all $\Phi_{i,+} = \frac{4 \Phi_+}{3n}$.
The following observation that can be found in~\cite{TalwarW13} 
\begin{equation}
\begin{aligned}
         \sum_{i\geq3n/4}p_i &\geq \frac{1}{4} + \varepsilon
\implies \sum_{i\leq3n/4}p_i &\leq 1 - \frac{1}{4} - \varepsilon=\frac{3}{4} - \varepsilon
\end{aligned}
\end{equation}
The result follows from combining these two facts.
\begin{equation}
     \sum_{i\leq\frac{3}{4}n}p_i\cdot\Phi_{i,+} 
\leq \left(\frac{3}{4}-\varepsilon\right)\frac{4\Phi_+}{3n}
=    \left(1-\frac{4\varepsilon}{3n}\right)\cdot\Phi_+
     \leq(1-2\varepsilon)\cdot\frac{\Phi_+}{n}
.
\end{equation}
Equation~\eqref{heavy14} follows similarly.
\end{proof}

\begin{claim}\label{clm:sumetaiphi}
Consider a round $t$ and a constant $\alpha\geq0$.
The following inequalities hold:
\begin{enumerate}
\item $\displaystyle\sum_{i\in\intcc{n}}\alpha\eta_i(\alpha\eta_i-1)\cdot\Phi_{i,+}(\bm{\load}(t))\leq\alpha^2\eta\nu\cdot\min\bigl(n,\Phi_+(\bm{\load}(t))\bigr)$.
\item $\displaystyle\sum_{i\in\intcc{n}}\alpha\eta_i(\alpha\eta_i+1)\cdot\Phi_{i,-}(\bm{\load}(t))\leq\alpha^2\eta\nu\cdot\Phi_-(\bm{\load}(t))$.
\end{enumerate}
\end{claim}
\begin{proof}
For the first statement, we calculate $\sum_{i\in\intcc{n}}\alpha\eta_i(\alpha\eta_i-1)\cdot\Phi_{i,+}(\bm{\load}(t))$
\begin{equation}
\begin{aligned}
&=    \sum_{i\leq\nu n}\alpha\eta_i(\alpha\eta_i-1)\cdot\Phi_{i,+}(\bm{\load}(t))+\sum_{i>\nu n}\alpha\eta_i(\alpha\eta_i-1)\cdot\Phi_{i,+}(\bm{\load}(t))\\
&=    \alpha\eta(\alpha\eta-1)\cdot\sum_{i\leq\nu n}\Phi_{i,+}(\bm{\load}(t))+\alpha\nu(1+\alpha\nu)\cdot\sum_{i>\nu n}\Phi_{i,+}(\bm{\load}(t))\\
&\leq \alpha\eta(\alpha\eta-1)\cdot\nu\cdot\Phi_+(\bm{\load}(t))+\alpha\nu(1+\alpha\nu)\cdot\eta\cdot\min\bigl(n,\Phi_+(\bm{\load}(t))\bigr)\\
&\leq \alpha^2\eta\nu\cdot\min\bigl(n,\Phi_+(\bm{\load}(t))\bigr)
,
\end{aligned}
\end{equation}
where the first inequality uses that $\Phi_{i,+}(\bm{\load}(t))$ is non-increasing in $i$ and that $\Phi_{i,+}(\bm{\load}(t))\leq1$ for all $i>\nu n$.
The claim's second statement follows by a similar calculation, using that $\Phi_{i,-}(\bm{\load}(t))$ is non-decreasing in $i$ (note that we cannot apply the same trick as above to get $\min\bigl(n,\Phi_-(\bm{\load}(t))\bigr)$ instead of $\Phi_-(\bm{\load}(t))$).
\end{proof}

\section{Missing Proofs for the 2-Choice Process}\label{app:missing:two_choice}

\begin{proof}[Proof of Observation~\ref{obs:two_choice:oneround_potdrop}]
Remember that $\mathbbm{1}_i$ is an indicator value which equals $1$ if and only if the $i$-th bin is non-empty in configuration $\bm{\load}$.
Bin $i$ looses exactly $\mathbbm{1}_i$ balls and receives exactly $k$ balls, such that $\load'_i-\load_i=-\mathbbm{1}_i+k$.
Similarly, we have $\varnothing'-\varnothing=-\nu+K/n$ for the change of the average load.
With the identity $\eta_i=\mathbbm{1}_i-\nu$ (see Section~\ref{sec:model}), this yields
\begin{equation}
\begin{aligned}
\Delta_{i,+}(t)
&= e^{\alpha\cdot\bigl(\load_i'-\varnothing'\bigr)}-e^{\alpha\cdot\bigl(\load_i-\varnothing\bigr)}\\
&= e^{\alpha\cdot\bigl(\load_i-\varnothing\bigr)}\cdot\left(e^{\alpha\cdot\bigl(-\mathbbm{1}_i+k+\nu-K/n\bigr)}-1\right)=\Phi_{i,+}\cdot\left(e^{\alpha\cdot(k-\eta_i-K/n)}-1\right)
,
\end{aligned}
\end{equation}
proving the first statement.
The second statement follows similarly.
\end{proof}

\begin{proof}[Proof of Lemma~\ref{lem:two_choice_potdrop_largepot}]
We prove the statement for $R=+$.
The case $R=-$ follows similarly.
Using Lemma~\ref{lem:pot_change} and summing up over all $i\in\intcc{n}$ we get
\begin{equation}
\begin{aligned}
\Ex{\Delta_+(t+1)|\bm{\load}}
&\leq \sum_{i\in\intcc{n}}\left(-\alpha\cdot(\eta_i+\hat{\delta}_i)+\alpha^2\cdot(\eta_i+\hat{\delta}_i)^2\right)\cdot\Phi_{i,+}\\
&=    \sum_{i\in\intcc{n}}\left(\eta_i\alpha(\eta_i\alpha-1)+\alpha^2\cdot(2\eta_i\hat{\delta}_i+\hat{\delta}_i^2)-\alpha\cdot\hat{\delta}_i\right)\cdot\Phi_{i,+}\\
&\leq \sum_{i\in\intcc{n}}\left(\eta_i\alpha(\eta_i\alpha-1)+5\alpha^2\lambda+\frac{5}{4}\alpha\lambda\right)\cdot\Phi_{i,+}
.
\end{aligned}
\end{equation}
Here, the last inequality uses $\lambda\leq1$ and $\abs{\hat{\delta}_i}\leq\frac{5}{4}\lambda$ (Claim~\ref{clm:two_choice:delta_bounds}).
We now apply Claim~\ref{clm:sumetaiphi}, $\nu\eta\leq1/4\leq\lambda$, and $\alpha<1/8$ to get
\begin{equation*}
\Ex{\Delta_+(t)|\bm{\load}}\leq\left(\alpha^2\lambda+5\alpha^2\lambda+\frac{5}{4}\alpha\lambda\right)\cdot\Phi_+<2\alpha\lambda\cdot\Phi_+
.
\qedhere
\end{equation*}
\end{proof}

\begin{proof}[Proof of Lemma~\ref{lem:bal_cfg:bound_upper_pot}]
To calculate the expected upper potential change, we use Lemma~\ref{lem:pot_change} and sum up over all $i\in\intcc{n}$ (using similar inequalities as in the proof of Lemma~\ref{lem:two_choice_potdrop_largepot} and the definition of $\hat{\delta}_i$):
\begin{equation}
\begin{aligned}
\Ex{\Delta_+(t+1)|\bm{\load}}
&\leq 6\alpha^2\lambda\cdot\Phi_+-\sum_{i\in\intcc{n}}\alpha\cdot\hat{\delta}_i\cdot\Phi_{i,+}\\
&=    \left(6\alpha^2\lambda-\alpha\lambda\cdot\check{1}\right)\cdot\Phi_++\hat{\alpha}\lambda n\sum_{i\in\intcc{n}}p_i\cdot\Phi_{i,+}
.
\end{aligned}
\end{equation}
We now use that $\Phi_{i,+}=e^{\alpha\cdot(\load_i-\varnothing)}\leq1$ for all $i>\frac{3}{4}n$ (by our assumption on $\load_{\frac{3}{4}n}$).
This yields
\begin{equation}
\Ex{\Delta_+(t+1)|\bm{\load}}\leq\left(6\alpha^2\lambda-\alpha\lambda\cdot\check{1}\right)\cdot\Phi_++\hat{\alpha}\lambda n\sum_{i\leq\frac{3}{4}n}p_i\cdot\Phi_{i,+}+2\alpha\lambda n
.
\end{equation}
Finally, we apply Claim~\ref{clm:bulk_leftright} and the definition of $\check{1}$ and $\hat{\alpha}$ to get
\begin{equation}
\begin{aligned}
\Ex{\Delta_+(t+1)|\bm{\load}}
&\leq\left(6\alpha^2\lambda-\alpha\lambda\cdot\check{1}+(1-2\varepsilon)\cdot\hat{\alpha}\lambda\right)\cdot\Phi_++2\alpha\lambda n\\
&\leq \left(4\alpha^2\lambda-2\varepsilon\cdot\alpha\lambda\right)\cdot\Phi_++2\alpha\lambda n
.
\end{aligned}
\end{equation}
Using $\alpha\leq\varepsilon/4$ yields the desired result.
\end{proof}

\begin{proof}[Proof of Lemma~\ref{lem:bal_cfg:bound_lower_pot}]
To calculate the expected lower potential change, we use Lemma~\ref{lem:pot_change} and sum up over all $i\in\intcc{n}$ (as in the proof of Lemma~\ref{lem:bal_cfg:bound_upper_pot}):
\begin{equation}
\begin{aligned}
\Ex{\Delta_-(t+1)|\bm{\load}}
&\leq 6\alpha^2\lambda\cdot\Phi_-+\sum_{i\in\intcc{n}}\alpha\cdot\check{\delta}_i\cdot\Phi_{i,-}\\
&=    \left(6\alpha^2\lambda+\alpha\lambda\cdot\hat{1}\right)\cdot\Phi_--\check{\alpha}\lambda n\sum_{i\in\intcc{n}}p_i\cdot\Phi_{i,-}
.
\end{aligned}
\end{equation}
We now use that $\Phi_{i,-}=e^{\alpha\cdot(\varnothing-\load_i)}\leq1$ for all $i\leq\frac{n}{4}$ (by our assumption on $\load_{\frac{n}{4}}$) and apply Claim~\ref{clm:bulk_leftright} to get
\begin{equation}
\begin{aligned}
\Ex{\Delta_-(t)|\bm{\load}}
&\leq \left(6\alpha^2\lambda+\alpha\lambda\cdot\hat{1}\right)\cdot\Phi_--(1+2\varepsilon)\cdot\check{\alpha}\lambda n\cdot\frac{\Phi_--\frac{n}{4}}{n}\\
&=    \left(6\alpha^2\lambda+\alpha\lambda\cdot\hat{1}-(1+2\varepsilon)\cdot\check{\alpha}\lambda\right)\cdot\Phi_-+(1+2\varepsilon)\cdot\frac{\check{\alpha}\lambda n}{4}\\
&\leq \left(8\alpha^2\lambda-2\varepsilon\cdot\alpha\lambda\right)\cdot\Phi_-+\frac{\alpha\lambda n}{2}
,
\end{aligned}
\end{equation}
where the last inequality used the definitions of $\hat{1}$, $\check{\alpha}$, as well as $\check{\alpha}>\alpha-\alpha^2$.
Using $\alpha\leq\varepsilon/8$ yields the desired result.
\end{proof}

\begin{proof}[Proof of Lemma~\ref{lem:potdrop_unbalanced_upper}]
Let $L\coloneqq\sum_{i\in\intcc{n}}\max(\load_i-\varnothing,0)=\sum_{i\in\intcc{n}}\max(\varnothing-\load_i,0)$ be the \enquote{excess load} above and below the average.
First note that the assumption $\load_{\frac{3}{4}n}\geq\varnothing$ implies $\Phi_-\geq\frac{n}{4}\cdot\exp(\frac{\alpha L}{n/4})$ (using Jensen's inequality).
On the other hand, we can use the assumption $\Ex{\Delta_+(t+1)|\bm{\load}}\geq-\frac{\varepsilon\alpha\lambda}{4}\cdot\Phi_+$ to show an upper bound on $\Phi_+$.
To this end, we use Lemma~\ref{lem:pot_change} and sum up over all $i\in\intcc{n}$ (as in the proof of Lemma~\ref{lem:bal_cfg:bound_upper_pot}):
\begin{equation}
\begin{aligned}
\Ex{\Delta_+(t+1)|\bm{\load}}
&\leq 6\alpha^2\lambda\cdot\Phi_+-\sum_{i\in\intcc{n}}\alpha\cdot\hat{\delta}_i\cdot\Phi_{i,+}\\
&=    6\alpha^2\lambda\cdot\Phi_+-\sum_{i\leq\frac{n}{3}}\alpha\cdot\hat{\delta}_i\cdot\Phi_{i,+}-\sum_{i>\frac{n}{3}}\alpha\cdot\hat{\delta}_i\cdot\Phi_{i,+}
.
\end{aligned}
\end{equation}
For $i\leq n/3$ we have $p_i=\frac{2i-1}{n^2}\leq\frac{2}{3n}$ and, using the definition of $\check{1}$ and $\hat{\alpha}$, $\hat{\delta}_i=\lambda n\cdot\bigl(\sfrac{1}{n}\cdot\check{1}-p_i\cdot\sfrac{\hat{\alpha}}{\alpha}\bigr)\geq(1-5\alpha)\lambda/3$.
Setting $\Phi_{\leq n/3,+}\coloneqq\sum_{i\leq n/3}\Phi_{i,+}$ and $\Phi_{>n/3,+}\coloneqq\sum_{i>n/3}\Phi_{i,+}$, together with Claim~\ref{clm:two_choice:delta_bounds} this yields
\begin{equation}
\begin{aligned}
\Ex{\Delta_+(t+1)|\bm{\load}}
&\leq 6\alpha^2\lambda\cdot\Phi_+-\frac{\alpha(1-5\alpha)\lambda}{3}\cdot\Phi_{\leq n/3,+}+\frac{5}{4}\alpha\lambda\cdot\Phi_{>n/3,+}\\
&=    \left(6\alpha^2\lambda-\frac{\alpha(1-5\alpha)\lambda}{3}\right)\cdot\Phi_++\left(\frac{5}{4}\alpha\lambda+\frac{\alpha(1-5\alpha)\lambda}{3}\right)\cdot\Phi_{>n/3,+}\\
&\leq -\frac{\varepsilon\alpha\lambda}{2}\cdot\Phi_++2\alpha\lambda\cdot\Phi_{>n/3,+}
,
\end{aligned}
\end{equation}
where the last inequality uses $\alpha\leq1/46\leq\frac{1}{23}-\frac{3}{46}\varepsilon$.
With this, the assumption $\Ex{\Delta_+(t+1)|\bm{\load}}\geq-\frac{\varepsilon\alpha\lambda}{4}\cdot\Phi_+$ implies $\Phi_+\leq\frac{8}{\varepsilon}\cdot\Phi_{>n/3,+}\leq\frac{8}{\varepsilon}\cdot\frac{2n}{3}e^{\frac{\alpha L}{n/3}}=\frac{16n}{3\varepsilon}e^{\frac{3\alpha L}{n}}$ (the last inequality uses that none of the $2n/3$ remaining bins can have a load higher than $L/(n/3)$).
To finish the proof, assume $\Phi_+>\frac{\varepsilon}{4}\cdot\Phi_-$ (otherwise the lemma holds).
Combining this with the upper bound on $\Phi_+$ and with the lower bound on $\Phi_-$, we get
\begin{equation}
\frac{16n}{3\varepsilon}e^{\frac{3\alpha L}{n}}\geq\Phi_+>\frac{\varepsilon}{4}\cdot\Phi_-\geq\frac{\varepsilon n}{16}\cdot e^{\frac{4\alpha L}{n}}
.
\end{equation}
Thus, the excess load can be bounded by $L<\frac{n}{\alpha}\cdot\ln\left(\frac{256}{3\varepsilon^2}\right)$.
Now, the lemma's statement follows from $\Phi=\Phi_++\Phi_-<\frac{5}{\varepsilon}\cdot\Phi_+\leq\frac{80n}{3\varepsilon^2}e^{\frac{3\alpha L}{n}}=\varepsilon^{-8}\cdot\LDAUOmicron{n}$.
\end{proof}

\begin{proof}[Proof of Lemma~\ref{lem:potdrop_unbalanced_lower}]
Let $L\coloneqq\sum_{i\in\intcc{n}}\max(\load_i-\varnothing,0)=\sum_{i\in\intcc{n}}\max(\varnothing-\load_i,0)$ be the \enquote{excess load} above and below the average.
First note that the assumption $\load_{\frac{n}{4}}\leq\varnothing$ implies $\Phi_+\geq\frac{n}{4}\cdot e^{\frac{\alpha L}{n/4}}$ (using Jensen's inequality).
On the other hand, we can use the assumption $\Ex{\Delta_-(t+1)|\bm{\load}}\geq-\frac{\varepsilon\alpha\lambda}{4}\cdot\Phi_-$ to show an upper bound on $\Phi_-$.
To this end, we use Lemma~\ref{lem:pot_change} and sum up over all $i\in\intcc{n}$ (as in the proof of Lemma~\ref{lem:bal_cfg:bound_lower_pot}):
\begin{equation}
\begin{aligned}
\Ex{\Delta_-(t+1)|\bm{\load}}
&\leq 6\alpha^2\lambda\cdot\Phi_-+\sum_{i\in\intcc{n}}\alpha\cdot\check{\delta}_i\cdot\Phi_{i,-}\\
&=    6\alpha^2\lambda\cdot\Phi_-+\sum_{i\leq\frac{2n}{3}}\alpha\cdot\check{\delta}_i\cdot\Phi_{i,-}+\sum_{i>\frac{2n}{3}}\alpha\cdot\check{\delta}_i\cdot\Phi_{i,-}
.
\end{aligned}
\end{equation}
For $i\geq2n/3$ we have $p_i=\frac{2i-1}{n^2}\geq\frac{4}{3n}-\frac{1}{n^2}$.
Using this with $p_i\leq p_n\leq2/n$ and $\check{\alpha}\geq\alpha-\alpha^2$, we can bound $\check{\delta}_i=\lambda n\cdot\bigl(\sfrac{1}{n}\cdot\hat{1}-p_i\cdot\sfrac{\check{\alpha}}{\alpha}\bigr)\leq\lambda\cdot(-\sfrac{1}{3}+\frac{1+\alpha}{n})+2\alpha\lambda\leq-\lambda/6+2\alpha\lambda$.
Setting $\Phi_{\leq2n/3,-}\coloneqq\sum_{i\leq2n/3}\Phi_{i,-}$ and $\Phi_{>2n/3,-}\coloneqq\sum_{i>2n/3}\Phi_{i,-}$, together with Claim~\ref{clm:two_choice:delta_bounds} this yields
\begin{equation}
\begin{aligned}
\Ex{\Delta_-(t+1)|\bm{\load}}
&\leq 6\alpha^2\lambda\cdot\Phi_-+\frac{5}{4}\alpha\lambda\cdot\Phi_{\leq2n/3,-}-\frac{\alpha\lambda}{6}\cdot\Phi_{>2n/3,-}+2\alpha^2\lambda\cdot\Phi_{>2n/3,-}\\
&\leq \left(8\alpha^2\lambda-\alpha\lambda/6\right)\cdot\Phi_-+\left(\frac{5}{4}\alpha\lambda+\alpha\lambda/6\right)\cdot\Phi_{\leq2n/3,-}\\
&\leq -\frac{\varepsilon\alpha\lambda}{2}\cdot\Phi_-+2\alpha\lambda\cdot\Phi_{\leq2n/3,-}
,
\end{aligned}
\end{equation}
where the last inequality uses $\alpha\leq1/32\leq\frac{1}{16}-\frac{1}{48}\varepsilon$.
With this, the assumption $\Ex{\Delta_-(t+1)|\bm{\load}}\geq-\frac{\varepsilon\alpha\lambda}{4}\cdot\Phi_-$ implies that $\Phi_-\leq\frac{8}{\varepsilon}\cdot\Phi_{\leq2n/3,-}\leq\frac{8}{\varepsilon}\cdot\frac{2n}{3}e^{\frac{\alpha L}{n/3}}=\frac{16n}{3\varepsilon}e^{\frac{3\alpha L}{n}}$ (the last inequality uses that none of the $2n/3$ remaining bins can have a load higher than $L/(n/3)$).
To finish the proof, assume $\Phi_->\frac{\varepsilon}{4}\cdot\Phi_+$ (otherwise the lemma holds).
Combining this with the upper bound on $\Phi_-$ and with the lower bound on $\Phi_+$, we get
\begin{equation}
\frac{16n}{3\varepsilon}e^{\frac{3\alpha L}{n}}\geq\Phi_->\frac{\varepsilon}{4}\cdot\Phi_+\geq\frac{\varepsilon n}{16}\cdot e^{\frac{4\alpha L}{n}}
.
\end{equation}
Thus, the excess load can be bounded by $L<\frac{n}{\alpha}\cdot\ln\left(\frac{256}{3\varepsilon^2}\right)$.
Now, the lemma's statement follows from $\Phi=\Phi_++\Phi_-<\frac{5}{\varepsilon}\cdot\Phi_-\leq\frac{80n}{3\varepsilon^2}e^{\frac{3\alpha L}{n}}=\varepsilon^{-8}\cdot\LDAUOmicron{n}$.
\end{proof}

\begin{proof}[Proof of lemma~\ref{lem:two_choice:onestep_phi_bound}]
The proof is via case analysis
\paragraph{Case 1: $x_\frac{n}{4} \geq \varnothing$ and $x_\frac{3n}{4} \leq \varnothing$}

In this case the desired bound follows from lemma~\ref{lem:bal_cfg:bound_upper_pot} and lemma~\ref{lem:bal_cfg:bound_lower_pot}.

\paragraph{Case 2: $x_\frac{n}{4} \geq x_\frac{3n}{4} \geq \varnothing$}
For $\Ex{\Delta_+(t+1)|\bm{\load}(t)} \leq \frac{-\varepsilon \alpha}{4} \Phi_+$ the results follows from lemma~\ref{lem:bal_cfg:bound_lower_pot}.

Recall $\Ex{\Delta(t+1)|\bm{\load}(t)} = \Ex{\Delta_+(t+1)|\bm{\load}(t)} + \Ex{\Delta_-(t+1)|\bm{\load}(t)}$

By lemma~\ref{lem:bal_cfg:bound_lower_pot} we can derive the following
\begin{equation}
\Ex{\Delta(t+1)|\bm{\load}(t)} \leq \Ex{\Delta_+(t+1)|\bm{\load}(t)} -\varepsilon\alpha\lambda\cdot\Phi_-(\bm{\load}(t))+\frac{\alpha\lambda n}{2}
\end{equation}
We now show that the RHS can be bounded by the RHS of lemma~\ref{lem:two_choice:onestep_phi_bound}. The result holds when
\begin{equation}
\begin{aligned}
\Ex{\Delta_+(t+1)|\bm{\load}(t)} 
&\leq -\frac{\varepsilon\alpha\lambda}{4}\cdot\Phi(\bm{\load}(t))+\varepsilon^{-8}\cdot\LDAUOmicron{n} +\varepsilon\alpha\lambda\cdot\Phi_-(\bm{\load}(t))-\frac{\alpha\lambda n}{2} 
\\
&\leq -\frac{\varepsilon\alpha\lambda}{4}\cdot\Phi(\bm{\load}(t)) +\varepsilon\alpha\lambda\cdot\Phi_-(\bm{\load}(t))
\\
&\leq -\frac{\varepsilon\alpha\lambda}{4}\cdot \left(\Phi_+(\bm{\load}(t)) + \Phi_-(\bm{\load}(t)) \right) +\varepsilon\alpha\lambda\cdot\Phi_-(\bm{\load}(t))
\\
&\leq -\frac{\varepsilon\alpha\lambda}{4}\cdot \Phi_+(\bm{\load}(t)) 
\end{aligned}
\end{equation}
For $\Ex{\Delta_+(t+1)|\bm{\load}(t)} \geq \frac{-\varepsilon \alpha}{4} \Phi_+$ lemma~\ref{lem:potdrop_unbalanced_upper} gives two subcases

\paragraph{Case 2.1 $ \Phi_+(\bm{\load}(t))\leq\frac{\varepsilon}{4}\cdot\Phi_-(\bm{\load}(t))$}
Using lemma~\ref{lem:two_choice_potdrop_largepot} and lemma~\ref{lem:bal_cfg:bound_lower_pot} we obtain the following

\begin{equation}
\begin{aligned}
\Ex{\Delta(t+1)|\bm{\load}(t)} &\leq 2\alpha\lambda \cdot \Phi_+(\bm{\load}(t)) -\varepsilon\alpha\lambda\cdot\Phi_-(\bm{\load}(t))+\frac{\alpha\lambda n}{2}
\\
&\leq \frac{2\alpha\lambda\varepsilon}{4} \cdot \Phi_-(\bm{\load}(t)) -\varepsilon\alpha\lambda\cdot\Phi_-(\bm{\load}(t))+\frac{\alpha\lambda n}{2}
\\
&= -\frac{\varepsilon\alpha\lambda}{2}  \cdot \Phi_-(\bm{\load}(t))+\frac{\alpha\lambda n}{2}
\\
&\leq -\frac{\varepsilon\alpha\lambda}{4}  \cdot \Phi(\bm{\load}(t)) +\varepsilon^{-8}\cdot\LDAUOmicron{n}
\end{aligned}
\end{equation}
\paragraph{Case 2.2 $\Phi(\bm{\load}(t))=\varepsilon^{-8}\cdot\LDAUOmicron{n}$}
Using lemma~\ref{lem:two_choice_potdrop_largepot} we get that $\Ex{\Delta(t+1)|\bm{\load}(t)} \leq 2\alpha\lambda \varepsilon^{-8}\cdot\LDAUOmicron{n}$

It remains to show that 
\begin{equation}
2\alpha\lambda\varepsilon^{-8} \cdot\LDAUOmicron{n} \leq -\frac{\varepsilon\alpha\lambda}{4}  \cdot \Phi(\bm{\load}(t)) +\varepsilon^{-8}\cdot\LDAUOmicron{n}
\end{equation}

Since $\Phi(\bm{\load}(t))=\varepsilon^{-8}\cdot\LDAUOmicron{n}$

\begin{equation}
2\alpha\lambda\varepsilon^{-8} \cdot\LDAUOmicron{n} \leq \left(1-\frac{\varepsilon\alpha\lambda}{4}\right)  \cdot \varepsilon^{-8} \cdot\LDAUOmicron{n}
\end{equation}

This holds where $2\alpha\lambda\leq\bigl(1-\frac{\varepsilon\alpha\lambda}{4}\bigr)$.
By definition $\alpha \leq 1/8$ and $\lambda < 1$.
The result follows. 

\paragraph{Case 3: $x_\frac{3n}{4} \leq x_\frac{ n}{4} \leq \varnothing$}
This case is similar to case 2.
For $\Ex{\Delta_-(t+1)|\bm{\load}(t)} \leq \frac{-\varepsilon \alpha n}{4} \Phi_-$ the results follows from lemma~\ref{lem:bal_cfg:bound_upper_pot}.
For $\Ex{\Delta_-(t+1)|\bm{\load}(t)} \geq \frac{-\varepsilon \alpha n}{4} \Phi_-$ two subcases are given by lemma~\ref{lem:potdrop_unbalanced_lower}. 
\paragraph{Case 3.1 $\Phi_-(\bm{\load}(t))\leq\frac{\varepsilon}{4}\cdot\Phi_+(\bm{\load}(t))$ }
The result follows from applying lemma~\ref{lem:potdrop_unbalanced_lower} and lemma~\ref{lem:two_choice_potdrop_largepot}.
\paragraph{Case 3.2 $\Phi(\bm{\load}(t))=\varepsilon^{-8}\cdot\LDAUOmicron{n}$}
This result follows from lemma~\ref{lem:two_choice_potdrop_largepot}.
\end{proof}

\end{document}